\documentclass{amsart}
\usepackage{times}
\usepackage{setspace}
\usepackage{amssymb,amsmath,amscd, amsthm,stmaryrd,verbatim, mathrsfs}

\usepackage{color}

\newtheorem{Thm}{Theorem}
\newtheorem*{thm}{Theorem XI.2.4 of Ref. \cite{FK94}}

\newtheorem{theorem}{Theorem}[section]

\newtheorem{Th}{Theorem}[section]
\newtheorem{Lem}[theorem]{Lemma}

\newtheorem{Prop}{Proposition}[section]
\newtheorem{Cor}{Corollary}[section]

\newtheorem{Definition}{Definition}
\newtheorem{Def}{Definition}[section]

\newtheorem{rmk}{Remark}[section]

\numberwithin{equation}{section}

\newcommand{\bb}{\mathbb}
\newcommand{\ms}{\mathscr}
\newcommand{\mr}{\mathrm}
\newcommand{\frk}{\mathfrak}
\newcommand{\tr}{\mr{tr}\,}

\usepackage{hyperref}
\begin{document}

\title{Generalized Kepler Problems I: Without Magnetic Charges}
\author{Guowu Meng}

\address{Department of Mathematics, Hong Kong University of Science and
Technology, Clear Water Bay, Kowloon, Hong Kong}
\curraddr{Institute for Advanced Study, Einstein Drive, Princeton, New Jersey
08540 USA}

\email{mameng@ust.hk}
\thanks{The author was supported by Qiu Shi Science and Technologies
Foundation.}

 
\date{May 23, 2011}











\maketitle
\begin{abstract}
For each simple euclidean Jordan algebra $V$ of rank $\rho$ and degree $\delta$, we introduce a family of classical dynamic problems. These dynamical problems all share the characteristic features of the Kepler problem for planetary motions, such as existence of Laplace-Runge-Lenz vector and hidden symmetry. After suitable quantizations, a family of quantum dynamic problems, parametrized by the nontrivial Wallach parameter $\nu$, is obtained. Here, $\nu\in{\mathcal W}(V):=\{k {\delta\over 2}\mid k=1, \ldots, (\rho-1)\}\cup((\rho-1){\delta\over 2}, \infty)$ and was introduced by N. Wallach to parametrize the set of nontrivial scalar-type unitary lowest weight representations of the conformal group of $V$.

For the quantum dynamic problem labelled by $\nu$, the bound state spectra is
$-{1/2\over (I+\nu{\rho\over 2})^2}$, $I=0$, $1$, \ldots  and its Hilbert space of bound states gives a new realization
for the afore-mentioned representation labelled by $\nu$. A few
results in the literature about these representations become more explicit and more refined. 

The Lagrangian for a classical Kepler-type dynamic problem introduced here is still of the simple form:
${1\over 2} ||\dot x||^2+{1\over r}$. Here, $\dot x$ is the velocity of a unit-mass particle moving on the space
consisting of $V$'s semi-positive elements of a fixed rank, and  $r$ is the inner product of $x$ with the identity element of $V$.

\end{abstract}


\section {Introduction}
The goal of this paper and its sequels is to develop a general mathematical theory for the
Kepler problem of planetary motions. As with many endeavors of developing a
general theory for a nice mathematical object, one usually starts with a
reformulation of this object from a new perspective which captures the essence
of such object --- a crucial step which often requires introducing new
concepts, and then work out all the necessary technical details. 

For example, choosing continuous functions in calculus as the nice mathematical
object, introducing the concept of open set and reformulating the notion of
continuous function in terms of open sets, one arrives at the much broader
notion of continuous map between topological spaces. For a sophisticated
example, choosing the Riemann integral for continuous functions over close intervals
as the nice mathematical object, reformulating it as the Lebesgue integral --- a
step which requires introducing the concepts of measurable space and measure,
one arrives at the general integration theory. For still another
sophisticated example, choosing the Euler number of closed oriented manifolds as the nice
mathematical object, reformulating it as the obstruction to the existence of
nowhere vanishing vector fields --- a step which requires introducing the
concepts of vector bundle and obstruction, one arrives at the theory
of characteristic classes. 

In our case, the Kepler problem of planetary motions is the nice mathematical
object, and we reformulate it as a dynamic problem on the {\it open future light
cone} 
$$
\{(t, \vec x)\in {\bb R}^4\mid t^2=|\vec x|^2, t>0\}
$$
of the Minkowski space \footnote{This reformulation seems to support the thesis of 2T-physics developed by I. Bars \cite{IBars10}.}. To quickly convince readers that this is a much
better formulation, one notes that the orbits (including the colliding ones) in this new formulation have a
simpler and uniform description \footnote{This description emerged soon after Ref. \cite{meng09} was written up.}: \begin{center}\emph{orbits are precisely the
intersections of 2D planes with the open future light cone}.\end{center} 
A closer examination reveals that it is the euclidean Jordan algebra structure, something that is more fundamental than the Lorentz structure, that plays a pivotal role in this new formulation. Once this point is realized, it is just a matter of time to work out the details for the general theory.  In this general theory, a key concept is that of {\bf canonical cone}, something that generalizes the notion of open future light cone.  With this key concept introduced,  it is not hard to describe what a generalized Kepler problem is. For example, a {\bf generalized classical Kepler problem} is just the dynamical problem for a single particle on a canonical cone $\mathcal C$ moving under a conserved force whose potential function is $-{1\over r}$. (Here, $r$ is the function on $\mathcal C$ such that $r(x)$ is the inner product of $x\in \mathcal C$ with the identity element of the euclidean Jordan algebra.)

One would expect to obtain the {\bf generalized quantum Kepler problems} by quantizing the generalized classical Kepler problems. However, due to the operator ordering problem, in general we don't know how to do the quantization, so we seem to be stuck. The way out of it is to demand that the hidden symmetry in the generalized classical Kepler problems be still present after quantization. Incidentally, that leads to a rediscovery of those unitary lowest modules mentioned in the abstract. 

As a comparison with the general integration theory,
{\it euclidean Jordan algebras} are like {\it measurable spaces}, {\it canonical cones} (cf. Subsection \ref{SS:Canonical Cones}) are like {\it measured spaces}, a {\it generalized Kepler problem} (cf. Section \ref{S: GKP}) on a canonical cone is like an
{\it integration theory} on a measured space. 

\subsection {A general remark on euclidean Jordan algebras} Euclidean Jordan
algebras were first introduced by P. Jordan \cite{PJordan33} to formalize
the algebra of physics observables in quantum mechanics. With E.
Wigner and J. von Neumann, Jordan \cite{JVW34} classified the finite dimensional
euclidean Jordan algebras: {\it every
 finite dimensional euclidean Jordan algebra is a direct sum of simple ones, and
the simple ones consist of four infinite series and one exceptional.} 
 
Although euclidean Jordan algebras are abandoned by physicists quickly, they
have been extensively studied at much more general settings by mathematicians
since 1950's. For an authoritative account of the history of Jordan algebras,
readers may consult the recent book by K. McCrimmon \cite{K. McCrimmon2004}. 
However, for our specific purpose, the book by J. Faraut and A. Kor\'{a}nyi 
\cite{FK94} is sufficient --- everything we really need either is already there
or can be derived from there.
 
It is helpful for us to view euclidean Jordan algebras as the super-symmetric
analogues of compact real Lie algebras. Just as compact real Lie algebras are
analogues of the infinitesimal gauge group for electromagnetism, euclidean
Jordan algebras are analogues of the infinitesimal space-time, but with a
more refined euclidean Jordan algebra structure hidden behind. In our view, {\em it is this more refined hidden
structure that is responsible for the existence of the
inverse square law}. Therefore, we would not be surprised if someday euclidean
Jordan algebras indeed play an indispensable role in the study of the
fundamental physics.

\subsection {A quick review of euclidean Jordan algebras}\label{ss:eJA}
Throughout this paper and its sequels we always assume that $V$ is a finite
dimensional euclidean Jordan algebra. That means that $V$ has both the inner product structure and the Jordan algebra (with an identity element $e$) structure on the underlying finite dimensional real vector space, such that, 
the two structures are compatible in the sense that
the Jordan multiplication by $u\in V$, denoted by $L_u$, is always self-adjoint
with respect to the
inner product on $V$.  We write the Jordan product of $u, v\in V$ as $uv$, so
$uv=L_uv$,  and the {\bf Jordan triple product} of $u, v, w\in V$ as $\{uvw\}$.
By definition, $\{uvw\}=S_{uv}w$, where
\begin{eqnarray}
S_{uv}=[L_u, L_v]+L_{uv},
\end{eqnarray}
i.e., 
\begin{eqnarray}
\{uvw\}=u(vw)-v(uw)+(uv)w. 
\end{eqnarray}
It is a fact that a simple euclidean Jordan algebra is uniquely
determined by its rank $\rho$ and degree $\delta$.

$V$ is a real Jordan algebra means that $V$ is a real commutative algebra such that 
\begin{eqnarray}
[L_u, L_{u^2}]=0, \quad u\in V.\nonumber
\end{eqnarray}
The inner product of $u, v\in V$, denoted by $\langle u\mid v\rangle$, is assumed to be the one such that the length of
$e$ is one: $||e||=1$.  We further assume that $V$ is simple, then the inner
product is unique. A simple computation shows that
\begin{eqnarray}
[S_{uv}, S_{zw}]=S_{\{uvz\}w}-S_{z\{vuw\}}.
\end{eqnarray}
Therefore, these $S_{uv}$ span a real Lie algebra --- the {\bf structure algebra} of
$V$, denoted by $\frk{str}(V)$ or simply $\frk{str}$. It is a fact that the {\bf
derivation algebra} of $V$, denoted by $\frk{der}(V)$ or simply $\frk{der}$, is
a  Lie subalgebra of $\frk{str}(V)$. Note that $\frk{der}$ --- a generalization of
$\frk{so}(3)$ --- is compact and $\frk{str}$ --- a generalization of
$\frk{so}(3,1)\oplus \bb R$ --- is reductive.

The generalization of $\frk{so}(4,2)$, denoted by $\frk{co}(V)$ or simply
$\frk{co}$, was independently discovered by Tits, Kantor, and
Kroecher  \cite {TKK60}. Like $\frk{so}(4,2)$, $\frk{co}(V)$ --- the {\bf conformal Lie
algebra} of $V$,  is a simple real Lie algebra. In this paper and its sequels, the
universal enveloping algebra of $\frk{co}$ shall be referred to as the {\bf TKK
algebra} and the corresponding simply connected Lie group shall be referred to
as the {\bf conformal group} and denoted by $\mr{Co}(V)$ or $\mr{Co}$.

As a vector space, $$\frk{co}(V)=V\oplus \frk{str}(V)\oplus V^*.$$
If we rewrite $u\in V$ as $X_u$ and $\langle v\mid\;\rangle\in V^*$ as $Y_v$,
then the TKK algebra is determined by the {\bf TKK commutation relations}: 
\begin{eqnarray}\label{TKKRel}
\fbox{$\begin{matrix}[X_u, X_v] =0, \quad [Y_u, Y_v]=0, \quad [X_u,
Y_v] = -2S_{uv},\cr\\ [S_{uv},
X_z]=X_{\{uvz\}}, \quad [S_{uv}, Y_z]=-Y_{\{vuz\}},\cr\\
[S_{uv}, S_{zw}] = S_{\{uvz\}w}-S_{z\{vuw\}}.
\end{matrix}$}
\end{eqnarray}
Here, $u,v,z,w\in V$.

\subsection{Classical realization of TKK algebras}
The total cotangent space $T^*V=V\times V^*$ is a symplectic space. Let
$\{e_\alpha\}$ be an orthonomal basis for $V$, with respect to which,  a  point
$(x, \langle \pi\mid\;\rangle)\in T^*V$ can be represented by its coordinate
$(x^\alpha, \pi^\beta)$. Then the basic Poisson bracket relations on $T^*V$ are
$$\{x^\alpha, \pi^\beta\}=\delta^{\alpha\beta}, \quad \{x^\alpha, x^\beta\}=0, \quad \{\pi^\alpha, \pi^\beta\}=0.$$

Introducing the moment functions
\begin{eqnarray}
{\mathcal S}_{uv} :=\langle S_{uv}(x)\mid \pi\rangle, \quad {\mathcal X}_u: =
\langle x\mid \{\pi u\pi \}\rangle, \quad {\mathcal Y}_v: =\langle x\mid
v\rangle
\end{eqnarray} on $T^*V$,  we shall show that, for any $u$, $v$, $z$ and $w$ in
$V$,
\begin{eqnarray}\label{PoissonR}
\left\{\begin{matrix}
\{\mathcal X_u, \mathcal X_v\}=0, \quad \{\mathcal Y_u, \mathcal Y_v\}=0, \quad
\{\mathcal X_u,
\mathcal Y_v\} = -2\mathcal S_{uv},\cr\\ \{\mathcal S_{uv},
\mathcal X_z\}=\mathcal X_{\{uvz\}}, \quad \{\mathcal S_{uv}, \mathcal
Y_z\}=-\mathcal Y_{\{vuz\}},\cr\\
\{\mathcal S_{uv}, \mathcal S_{zw}\} = \mathcal S_{\{uvz\}w}-\mathcal
S_{z\{vuw\}}.
\end{matrix}\right.
\end{eqnarray}
The realization of $O$ by $\mathcal O$ is referred to as the classical realization of the TKK algebra. The quantization of this classical
realization leads to operator realizations of the TKK algebra.

\subsection{Operator realizations of TKK algebras}
The canonical quantization involves promoting classical physical variables $\mathcal O$ to
differential operators $\acute O$ (or the duals
$\grave O$) using recipe: $\pi^\alpha\to -i{\partial \over \partial
x^\alpha}$ (or $x^\alpha\to i{\partial \over \partial \pi^\alpha}$). Here is a
word of warning:  in order to get anti-hermitian differential operators in the end, instead of
using the quantized differential operators, we actually use the quantized differential operators multiplied
by $-i$.

Due to the ambiguity with the operator ordering, the canonical quantization has
an ambiguity, measured by a real parameter $\mu$ here. For simplicity, we write
$\sum_\alpha e_\alpha {\partial \over \partial x^\alpha}$ as ${/\hskip -5pt
\partial}$ and $\sum_\alpha e_\alpha {\partial \over \partial \pi^\alpha}$ as
${\backslash \hskip -6pt \partial}$. The quantization recipe, with the operator
orderings taken into account, yields either differential operators on $V$:

\begin{eqnarray}\label{explicitFor}
\left\{
\begin{array}{rcl}
{\acute S}_{uv}(\nu) &:=&-\langle S_{uv}(x)\mid {/\hskip -5pt
\partial}\rangle-{\nu\over 2}\tr (uv),\\
\\
{\acute X}_u (\nu)&:=& i \langle x\mid \{{/\hskip -5pt \partial} u {/\hskip -5pt
\partial} \}\rangle+i\nu\tr(u{/\hskip -5pt \partial}), \\
\\
{\acute Y}_v(\nu) &:=& -i\langle x\mid v\rangle
\end{array}\right.
\end{eqnarray}
or differential operators on $V^*$: 
\begin{eqnarray}\label{explicitFor*}
\left\{
\begin{array}{rcl}
{\grave  S}_{uv}(\nu) &:=& \langle S_{vu}(\pi) \mid  {\backslash    
 \hskip -6pt \partial}\rangle-{\nu^*\over 2}\tr (uv),\\
\\
{\grave  X}_u (\nu)&:=&   \langle \{\pi u \pi\} \mid {\backslash    
 \hskip -6pt \partial}\rangle -\nu^*\tr(u\pi), \\
\\
{\grave Y}_v(\nu) &:=& \langle v\mid {\backslash    
 \hskip -6pt \partial}\rangle.
\end{array}\right.
\end{eqnarray}
Here, $\nu^*=\nu-{2n\over \rho}$ with $n=\dim V$.

Eqns (\ref{explicitFor}) and (\ref{explicitFor*}) provide us two families of operator realizations for the TKK algebra. It is not hard to see that the two families are related by the Fourier transform and for the realizations to be unitary we must have $\nu\ge
0$.  Note that the case $\nu^*=0$ is well-known to physicists, cf. Ref. \cite{MG93}; moreover,
 for $\nu>1+(\rho-1)\delta$, the explicit formula in Eq.
(\ref{explicitFor}) has been obtained in Ref.  \cite{ADO06} via an indirect route.  

These operator realizations as given in Eq. (\ref{explicitFor}) are not unitary with respect to the obvious $L^2$-inner product
$$
(\psi_1, \psi_2)=\int_V \bar \psi_1\psi_2\, dm,
$$
where $dm$ is the Lebesgue measure. Therefore, the right hermitian inner products  must be found in order to have unitary realizations. For us, the clues come out naturally in our \cite{meng09, meng10} study of the Kepler problem.

\subsection{Quantum realizations of TKK algebras}\label{SS: QTKK}
In order to address the unitarity problem for the quantization given by the
explicit formula in Eq. (\ref{explicitFor}), we are led to introduce the notion
of {\bf canonical cones}. Let $k>0$ be an integer no more than $\rho$.  The
canonical cone of rank $k$, denoted by ${\mathcal C}_k$, is an open Riemannian
manifold. As a smooth manifold, it consists of all semi-positive elements in $V$
of rank $k$. The Riemannian metric on $\mathcal C_k$ is
defined in terms of Jordan multiplication, so it is invariant under the automorphism group of $V$. One important fact here is that
the operators defined in Eq. (\ref{explicitFor}) always descend to differential operators on $\mathcal C_k$.

Denote by $\mathcal P(V)$ the algebra of polynomial maps from $V$ to $\bb C$, by
$\mathcal P(\mathcal C_k)$ the algebra of functions on $\mathcal C_k$ coming from
the restriction of
elements in $\mathcal P(V)$, i.e.,
\begin{eqnarray}
\mathcal P(\mathcal C_k)=\{p: {\mathcal C_k} \to \bb C\mid p\in \mathcal P(V)\}.
\end{eqnarray}
We use $r$ to denote this function on $V$: 
\begin{eqnarray}\label{r-function}
r(x)=\langle e\mid x\rangle.
\end{eqnarray}
We shall show that $e^{-r}{\mathcal P}(\mathcal C_k)$ is a common domain for
the operators defined in Eq. (\ref{explicitFor}) and is dense in $L^2(d\mu_\nu)$, here $d\mu_\nu$ is a measure on $\mathcal C_\rho$
(on $\mathcal C_k$) if $\nu>(\rho-1){\delta\over 2}$ (if $\nu=k{\delta\over
2}$, $k=1$, $2$, \ldots, $\rho-1$), and can be explicitly written down
\footnote{The explicit formula for $d\mu_\nu$ has already been given in Ref.
\cite{FK94} when $\nu>(\rho-1){\delta\over 2}$.} in terms of
the volume form $\mr{vol}$ on the canonical cone. It shall also be demonstrated
that $L^2(d\mu_\nu)$
is a unitary lowest weight representation $\pi_\nu$ for $\mr{Co}$ with the
lowest weigh equal to $\nu\lambda_0$, here $\lambda_0$ is the fundamental weight
conjugate to the unique non-compact simple root of $\frk{co}$ under a suitable
choice of Cartan subalgebra. Note that, we have a unitary lowest weight
representation for $\mr{Co}$ when either $\nu=k{\delta\over 2}$ with $k=1,
\ldots, \rho-1$ or $\nu>(\rho-1){\delta\over 2}$, and these $\pi_\nu$ exhaust
all nontrivial scalar-type unitary lowest weight representations for $\mr{Co}$, according to
Ref. \cite{EHW82}. 

As a historical note, we would like to mention that the whole list of these representations $\pi_\nu$ was first discovered by N. Wallach \cite{Wallach79} via algebraic methods and the first analytic model for these representations was discovered by M. Vergne and H. Rossi \cite{Vergne&Rossi76}. In relatively recent years, these analytic models were further extended to real semi-simple Jordan algebras by A. Dvorsky and S. Sahi \cite{Dvorsky&Sahi03}. 

Even though we get the explicit formula for measure $d\mu_\nu$ via the
Riemannian metric  on canonical cones,  as far as representations are concerned,
the introduction of this Riemannian metric is not needed. However, it is essential if one wishes to go one step further, i.e., unravel the hidden dynamic models behind these unitary lowest weight representations.  

\subsection {Universal Kepler hamiltonian and universal Lenz vector}
The hidden dynamical models referred to in the proceeding paragraph are of the
Kepler type in the sense that the hidden conserved Lenz vector always exists.  To see
that,  we need to introduce the {\bf universal hamiltonian}
\begin{eqnarray}\label{universalH}
H:={1\over 2} Y_e^{-1} X_e+iY_e^{-1}
\end{eqnarray}
and the {\bf universal Lenz vector}
\begin{eqnarray}
A_u:=iY_e^{-1}[L_u, Y_e^2H], \quad u\in V.
\end{eqnarray}
Here, $i$ is the imaginary unit, and $Y_e^{-1}$ is the formal inverse of
$Y_e$.  Note that both $H$ and $A_u$ are elements of the
complexified TKK algebra with $Y_e$ inverted. 

The manifest conserved quantity is $L_{u,v}:=[L_u, L_v]$ --- the analogue of
angular momentum, and the hidden conserved quantity is $A_u$ --- the analogue of
original Lenz vector; in fact,  the following commutation relations have been
verified in Ref. \cite{meng10}:
\begin{eqnarray}
\left\{\begin{array}{lcl} 
[L_{u,v}, H] & = & 0\;,
\cr
[A_u, H] & = & 0\;,\cr
[L_{u,v},  L_{z, w}] & = &  L_{[L_u, L_v]z,w}+ L_{z, [L_u, L_v]w}\;,\cr
[L_{u,v}, A_z] &= & A_{[L_u, L_v]z}\;,\cr[A_u, A_v] & = &
-2H L_{u,v}\;.
\end{array}\right.
\end{eqnarray}

Now,  under a unitary lowest weight representation $\pi_\nu$ in the proceeding
subsection, both $H$ and $A_u$ become differential operators $\acute
H(\nu)$ and $\acute A_u(\nu)$ on the canonical cone. However, $\acute H(\nu)$ is
not quite right, because the term $-\acute Y_e^{-1}(\nu) \acute X_e(\nu)$ in
$\acute H(\nu)$ is not the Laplace operator on the
canonical cone, even
up to an additive function.

\subsection{Generalized Kepler problems}
By comparing with the Laplace operator on the canonical cone, one realizes that
$-\acute Y_e^{-1}(\nu) \acute X_e(\nu)$ is almost right: after conjugation by
the multiplication with a positive function on the canonical cone, modulo an
additive function, it becomes the Laplace operator on the canonical
cone. After this conjugation, $\acute O$ becomes a new differential
operator (shall be denoted by
$\tilde O$), and $L^2(d\mu_\nu)$ becomes $L^2({1\over r}\mr{vol})$.
Now $\tilde H(\nu)$ is the hamiltonian for the quantum Kepler-type dynamic problem
behind representation $\pi_\nu$. It is in this
sense we say that $L^2({1\over r}\mr{vol})$ is more natural than $L^2(d\mu_\nu)$.

The bound state spectrum for $\tilde H(\nu)$
is $$-{1/2\over (I+\nu{\rho\over 2})^2}, \quad I=0, 1, 2,\ldots $$
and the Hilbert space of bound states --- a closed subspace
of $L^2(\mr{vol})$, denoted by $\ms H(\nu)$, provides a new realization for
representation $\pi_\nu$.

The two natural realizations for $\pi_\nu$, one by $L^2({1\over r}\mr{vol})$,
one by $\ms H(\nu)$, implies that
there are two kinds of orthogonalities for generalized Laguerre polynomials,
generalizing the well-known fact that there are two kinds of orthogonalities for Laguerre
polynomials, cf. Ref. \cite{CD2003}. We would like to point out that this phenomena was first observed by A.O. Barut and H. Kleinert \cite{BarutKleinert67} for the original Kepler/Coulomb problem. 

By the way, one may arrive at a quantum Kepler-type dynamic problem by suitably quantizing a classical Kepler-type problem whose configuration space is a canonical cone and Lagrangian is
\begin{eqnarray}
L(x,\dot x)={1\over 2} ||\dot x||^2+{1\over r}.
\end{eqnarray} 
\subsection{Outline of the paper}
In Section \ref{S:Jordan Algebras}, we give a quick review of Jordan
algebras and especially euclidean Jordan algebras.  In Section \ref{Realizations of TKK algebras}, we first review TKK algebras, then we introduce the classical realization for TKK algebras. 
The operator realizations follows from formally quantizing the classical realization. Section \ref{quantization} is the most technical section of this article, here, the unitarity of those operator realizations is addressed. In Section \ref{S: GKP} we present the generalized Kepler problem behind each nontrivial unitary representation in Section \ref{quantization}, then we solve the bound state problem and prove that the Hilbert space of bound states provides another realization for the representation. Here, the proof is complete only after one proves an analogue of Theorem 2 in Ref. \cite{CD2003} for generalized Laguerre polynomials, something that can surely be done.  Two appendixes are given in the end. In Appendix \ref{App: A}, a technical theorem extending the existing one in Ref. \cite{FK94} is proved. In Appendix \ref{App:B}, for the convenience of readers, some basic notations for this paper are listed. 

\vskip 10pt
Ref. \cite{FK94} is our most consulted book --- really a bible for our
purpose.  However, we should warn the readers that there are some convention/notation differences: 1) our invariant inner product on $V$ is the one such
that the identity element $e$ has unit length and is denoted by $\langle \mid
\rangle$ rather than $(, )$, 2) the rank and degree of a Jordan algebra is
denoted by $\rho$ and $\delta$ rather than $r$ and $d$ because $r$ is reserved
for function $x\mapsto \langle e \mid x \rangle$ and $d$ is reserved for the
exterior derivative operator, 3) the Jordan multiplication by $u$ is denoted
by $L_u$ rather than $L(u)$, 4) the Jordan triple product of $u$, $v$ and $w$ is
denoted by $\{uvw\}$ rather than $\{u,v,w\}$, 5) we use $S_{uv}$ rather than
$u\mbox{\tiny{$\Box$}} v$ for endomorphism $w\mapsto \{uvw\}$. 

\vskip 10pt
{\bf Acknowledgment}: This work is performed while the author is visiting the {\it Institute for Advanced Study} (IAS) during the academic year 2010-2011. He would like to thank the {\it Qiu Shi Science and Technologies Foundation} for providing the financial support for his visit to the IAS. He would also like to thank Professors M. Atiyah, R. Howe and C. Taubes for continuous encouragements and the IAS faculty members P. Goddard, T. Spencer, P. Sarnak and E. Witten for enlightening discussions.

\section{Jordan Algebras}\label{S:Jordan Algebras}
Jordan algebra has become a big subject now; however, what is relevant for us here is
very minimal. In this section, we shall review
the bare essential of Jordan algebras, oriented towards our purpose. Apart
from Lemma \ref{LemmaVogan}, everything else presented here can be found from the book by J. Faraut and A.
Kor\'{a}nyi \cite{FK94}.

\subsection{Basic definitions}

Recall that an {\bf algebra} $V$ over a field $\bb F$ is a vector
space over $\bb F$ together with a $\bb F$-bilinear map $V\times
V\to V$ which maps $(u, v)$ to $uv$. This $\bb F$-bilinear map can
be recast as a linear map $V\to \mr{End}_{\bb F}(V)$ which maps $u$
to $L_u$: $v\mapsto uv$.

We say that algebra $V$ is {\it commutative} if $uv=vu$ for any $u,
v\in V$. As usual, we write $u^2$ for $uu$ and $u^{m+1}$ for $uu^m$
inductively.

\begin{Def} A {\bf Jordan algebra} over $\bb F$ is just a
commutative algebra $V$ over $\bb F$ such that
\begin{eqnarray}\label{JA}
[L_u, L_{u^2}]=0
\end{eqnarray} for any $u\in V$.
\end{Def}

As the first example, we note that $\bb F$ is a Jordan algebra over
$\bb F$. Here is a {\it recipe} to produce Jordan algebras. Suppose
that $\Phi$ is an associative algebra over field $\bb F$ with
characteristic $\neq 2$, and $V\subset \Phi$ is a linear subspace of
$\Phi$, closed under square operation, i.e, $u\in V$ $\Rightarrow$
$u^2\in V$. Then $V$ is a Jordan algebra over $\bb F$ under the
Jordan product
$$
uv:={(u+v)^2-u^2-v^2\over 2}.
$$
Applying this recipe, we have the following Jordan algebras over
$\bb R$:
\begin{enumerate}
\item The algebra $\Gamma(k)$. Here $\Phi=\mr{Cl}(\bb R^k)$---the Clifford
algebra of $\bb
R^k$ and $V=\bb R\oplus \bb R^k$.

\item The algebra ${\mathcal H}_k(\bb R)$. Here $\Phi=M_k(\bb R)$---the algebra
of real $k\times k$-matrices and $V\subset \Phi$ is the set of symmetric
$k\times k$-matrices.

\item The algebra ${\mathcal H}_k(\bb C)$. Here $\Phi=M_k(\bb C)$---the algebra
of complex $k\times k$-matrices (considered as an algebra over $\bb R$)
and $V\subset \Phi$ is the set of Hermitian $k\times k$-matrices.

\item The algebra ${\mathcal H}_k(\bb H)$. Here $\Phi=M_k(\bb H)$---the algebra
of quaternionic $k\times k$-matrices (considered as an algebra over $\bb R$)
and $V\subset \Phi$ is the set of Hermitian $k\times k$-matrices.
\end{enumerate}
The Jordan algebras over $\bb R$ listed above are {\em
special} in the sense that they are derived from associated
algebras via the above recipe. Let us use ${\mathcal H}_k(\bb O)$ to denote the algebra for which
the underlying real vector space is the set of Hermitian
$k\times k$-matrices over $\bb O$ and the  product is the
symmetrization of the matrix product. One can show that ${\mathcal
H}_k(\bb O)$ is a Jordan algebra if and only if $k\le 3$.  However, only
${\mathcal
H}_3(\bb O)$ is new because ${\mathcal H}_1(\bb O)={\mathcal H}_1(\bb R)$,
${\mathcal
H}_2(\bb O)\cong \Gamma(9)$. ${\mathcal H}_3(\bb O)$  is called {\it
exceptional} because
it cannot be obtained via the above recipe,  even if the associative
algebra $\Phi$ is allowed to be infinite dimensional.

\subsection{Euclidean Jordan algebras}
Any Jordan algebra $V$ comes with a canonical symmetric
bilinear form \begin{eqnarray}\tau(u, v):=\mbox{the trace of
$L_{uv}$}.\end{eqnarray}
It is a fact that $L_u$ is self-adjoint with respect to $\tau$.

We say that Jordan algebra $V$ is {\it semi-simple} if the symmetric
bilinear form $\tau$ is non-degenerate. We say that Jordan algebra
$V$ is {\it simple} if it is semi-simple and has no ideal other than
$\{0\}$ and $V$ itself.

By definition, an {\bf euclidean Jordan algebra}\footnote{Called
{\it formally real Jordan algebra} in the old literature.} is a
real Jordan algebra with an identity element $e$ such that the symmetric bilinear form $\tau$ is
positive definite. Therefore, an euclidean Jordan algebra is
semi-simple and can be uniquely written as the direct sum of simple
ideals --- ideals which are simple as Jordan algebras.

\begin{Th}[Jordan, von Neumann and Wigner, Ref. \cite{JVW34}]\label{JAclassification}
The complete list of simple euclidean Jordan algebras are
\begin{enumerate}
\item The algebra $\Gamma(k)=\bb R\oplus \bb R^k$ ($k\ge 2$).

\item The algebra ${\mathcal H}_k(\bb R)$ ($k\ge 3$ or $k=1$).

\item The algebra ${\mathcal H}_k(\bb C)$ ($k\ge 3$).

\item The algebra ${\mathcal H}_k(\bb H)$ ($k\ge 3$).

\item The algebra ${\mathcal H}_3(\bb O)$.
\end{enumerate}
\end{Th}
Note that $\Gamma(1)$ is not simple and ${\mathcal H}_1(\bb F)$ ($=\bb R$) is the
only associative simple euclidean Jordan algebra. Note also that there are
various isomorphisms: $\Gamma(2)\cong {\mathcal H}_2(\bb R)$,
$\Gamma(3)\cong {\mathcal H}_2(\bb C)$, $\Gamma(5)\cong {\mathcal
H}_2(\bb H)$, $\Gamma(9)\cong {\mathcal H}_2(\bb O)$.

\begin{rmk}
A simple euclidean Jordan algebra is uniquely specified by its rank $\rho$ and
degree $\delta$:
\begin{displaymath}
\begin{array}{|c|c|c|c|c|c|}
\hline
J & \Gamma(k) & {\mathcal H}_k(\bb R) & {\mathcal H}_k(\bb C) &{\mathcal
H}_k(\bb H) & {\mathcal H}_3(\bb O)\\
\hline
\rho & 2 & k &  k & k &  3\\
\hline
\delta & k-1 & 1 &  2 &  4 &  8\\
\hline
\end{array}
\end{displaymath}
Hence, for the
simple euclidean Jordan algebras, there is one with rank-one,
infinity many with rank two, four with rank three, and three with
rank four or higher.
\end{rmk}

The notion of {\bf trace} is valid for Jordan algebras. For the
simple euclidean Jordan algebras, the trace can be easily described:
For $\Gamma(k)$, we have
$$
\tr(\lambda, \vec u)=2\lambda,
$$
and for all other types, it is the usual
trace for matrices.

For the {\bf inner product} on $V$, we take
\begin{eqnarray}\label{innerProduct}
\fbox{$\langle u\mid v\rangle :={1\over
\rho}\tr(uv)$}
\end{eqnarray} so that the identity element $e$ becomes a unit vector. One
can check that $L_u$ is self-adjoint with respect to this inner
product: $ \langle v u\mid w\rangle = \langle v\mid u w\rangle$,
i.e., $L_u'=L_u$. For $u$, $v$ in $V$, we introduce linear map $S_{uv}:= [L_u,
L_v]+L_{uv}$, and write
$S_{uv}(w)$ as $\{uvw\}$. One can check that $S_{uv}'=S_{vu}$ and 
$$
[S_{uv}, S_{zw}]=S_{\{uvz\}w}-S_{z\{vuw\}}.
$$

\vskip 10pt
In the remainder of this paper, we fix a simple euclidean Jordan algebra $V$,
and use $e$, $\rho$, $\delta$ and $n$ to denote its identity element, rank,
degree and dimension respectively. We shall use $\{e_{ii}\}$  to denote a {\bf Jordan frame}  and $V_{ij}$ to denote
the resulting $(i,j)$-{\bf Pierce component}. 
Choose an orthogonal basis $e_{ij}^\mu$ for each $V_{ij}$ ($1\le i<j\le \rho$) with each basis vector having length $1\over \sqrt \rho$, here $1\le\mu\le \delta$. Then $$\{e_{kk}, e_{ij}^\mu\mid 1\le k\le \rho, 1\le i<j\le \rho, 1\le\mu\le\delta\}$$ is an orthogonal basis for $V$, and each basis vector has length $1\over \sqrt \rho$. Such a basis is referred to as a {\bf Jordan basis} with respect to the Jordan frame $\{e_{ii}\}$. 

Here is a convention we shall adopt: $x$ is reserved for a point in
the smooth space $V$, and $u$, $v$, $z$, $w$ are reserved for vectors in vector
space $V$. We shall also use $V$ to denote the Euclidean space
with underlying smooth space $V$ and Riemannian metric $ds^2_E$: 
\begin{eqnarray}
T_xV\times T_xV & \to &  \bb R \cr ((x, u), (x, v)) &\mapsto &
\langle u\mid v\rangle.
\end{eqnarray}

Finally, we would like to remark that, if one takes Jordan algebras as
analogues of Lie algebras, then euclidean Jordan algebras are the analogues of
compact Lie algebras.

\subsection{Structure algebras and TKK Algebras}\label{SS: TKK}
As always, we let $V$ be a finite dimensional simple euclidean Jordan algebra. We use $\Omega$ to denote
the {\bf symmetric cone} of $V$ and $\mr{Str}(V)$ to denote the {\bf structure group} of $V$.  By definition,
$\Omega$ is the topological interior of $$\{x^2\mid x\in V\}$$ and 
$$
\mr{Str}(V)=\{g\in GL(V)\mid P(gx) =gP(x)g' \quad \forall x\in V\}.
$$
Here $P(x):=2L_x^2-L_{x^2}$ and it is called the {\bf quadratic representation} of $x$.

We write $V^{\bb C}$ for the complexification of $V$, denote by $T_\Omega$ the tube domain associated with $V$. By definition, $T_\Omega=V\oplus i\Omega$. We say that map $f$: $T_\Omega\to T_\Omega$ is a holomorphic automorphism of $T_\Omega$ if $f$ is invertible and both $f$ and $f^{-1}$ are holomorphic. We use $\mr{Aut}(T_\Omega)$ to denote the group of holomorphic automorphisms
of $T_\Omega$. 

It is a fact that both $\mr{Str}(V)$ and $\mr{Aut}(T_\Omega)$  are Lie groups.  The Lie algebra of $\mr{Str}(V)$ is referred to as the {\bf structure algebra} of $V$ and is denoted by $\frk{str}(V)$ or simply $\frk{str}$, the Lie algebra of
$\mr{Aut}(T_\Omega)$ is referred to as the {\bf conformal algebra} of $V$ and is denoted by $\frk{co}(V)$ or simply $\frk{co}$, and its universal enveloping algebra is called the {\bf TKK algebra} of $V$. The simply connected Lie group with $\frk{co}$ as its Lie algebra, denoted by $\mr{Co}(V)$ or simply $\mr{Co}$, shall be referred to as the {\bf conformal group} of $V$. 

Both the structure algebra and the conformal algebra have a simple direct algebraic description, cf. Subsection \ref{ss:eJA}. While the structure algebra is reductive, the conformal algebra is simple.
\vskip 10pt
Since $\frk{co}$ is a non-compact real simple Lie algebra, it
admits a Cartan involution $\theta$, unique up
to conjugations by inner automorphisms. Indeed, one can choose $\theta$ such
that
$$
\theta(X_u)=Y_u, \quad \theta(Y_u)=X_u, \quad
\theta(S_{uv})=-S_{vu}.
$$  The resulting Cartan decomposition is
$\frk{co}=\frk{u}\oplus\frk{p}$ with
$$\frk{u}=\mr{span}_{\bb R}\{[L_u, L_v], X_w+Y_w\mid u, v, w\in V\}, \quad
\frk{p}=\mr{span}_{\bb R}\{L_u, X_v-Y_v\mid u, v\in V\}.$$
Note that $\frk{u}$ is reductive with center spanned by $X_e+Y_e$ and its semi-simple part  $\bar {\frk{u}}$ is
$$ \mr{span}_{\bb R}\{[L_u, L_v], X_w+Y_w\mid u,
v, w\in ({\bb R}e)^\perp\}.$$
Sometime we need to emphasize the dependence on $V$, then we rewrite $\frk u$ as ${\frk u}(V)$. It is a fact that $\frk{str}$ and $\frk{u}$ are different real forms of the same complex reductive Lie algebra. In fact, one can identify their complexfications as follows:
\begin{eqnarray}\label{id: str=k}
[L_u, L_v] \leftrightarrow [L_u, L_v],\quad
-{i\over 2}(X_w+Y_w)  \leftrightarrow  L_w.
\end{eqnarray}
Here is a detailed summary of all real Lie algebras we have
encountered:
\begin{displaymath}
\begin{array}{|c|c|c|c|c|}
\hline
V & \frk{der} & \frk{str} & \frk{u} & \frk{co} \\
\hline
\Gamma(n) & \frk {so}(n) & \frk{so}(n,1)\oplus \bb R & \frk{so}(n+1)\oplus \frk{so}(2) & \frk{so}(n+1,2) \\
\hline
{\mathcal H}_n(\bb R)  & \frk{so}(n) & \frk{sl}(n,{\bb R})\oplus \bb R & \frk{u}(n) &  \frk{sp}(n,{\bb R}) \\
\hline
{\mathcal H}_n(\bb C)  & \frk{su}(n) & \frk{sl}(n,{\bb C})\oplus \bb R & \frk{su}(n)\oplus\frk{su}(n)\oplus \frk{u}(1) &  \frk{su}(n,n)\\
\hline
{\mathcal H}_n(\bb H) & \frk{sp}(n) &  \frk{su}^*(2n)\oplus \bb R & \frk{u}(2n)& \frk{so}^*(4n)  \\
\hline
{\mathcal H}_3(\bb O) & \frk{f}_4 & \frk{e}_{6(-26)}\oplus \bb R & \frk{e}_{6}\oplus \frk{so}(2) & \frk{e}_{7(-25)}  \\
\hline
\end{array}
\end{displaymath}

Recall that $e_{11}$ denotes the first element of a Jordan frame for $V$. The following lemma has been proved in Subsection 4.2 of Ref. \cite{meng09}.

\begin{Lem}\label{LemmaVogan}
There is a maximally compact $\theta$-stable Cartan subalgebra $\frk h$ for
$\frk{co}$, with respect to which, there is a simple root system
consisting of imaginary roots $\alpha_0$, $\alpha_1$, \ldots,
$\alpha_r$ such that, for $i\ge 1$, $\alpha_i$ is compact with $H_{\alpha_i},
E_{\pm\alpha_i}\in \bar { \frk u}^{\bb C}$,  and $\alpha_0$
is non-compact with
$$
H_{\alpha_0}=i(X_{e_{11}}+Y_{e_{11}}), \quad E_{\pm \alpha_0}={i\over
2}(X_{e_{11}}-Y_{e_{11}})\mp L_{e_{11}}.
$$  
\end{Lem}

\subsection{The decomposition of the action of $\mr{Str}(V)$ on $\mathcal P(V)$}\label{SS: decomposition}
The structure group acts on $V$ linearly, so it acts on  $\mathcal P(V)$ --- the set of complex-valued polynomial functions on $V$. The goal here is to describe the known decomposition of the action of $\mr{Str}(V)$ on $\mathcal P(V)$ into irreducible components.

With a Jordan frame $\{e_{ii}\}$ for $V$ chosen, for $1\le k\le \rho$, we let $e[k]=e_{11}+\cdots+e_{kk}$. Denote by $V_k$ the eigenspace of $L_{e[k]}$ with eigenvalue $1$ and by $P_k$ the orthogonal projection of $V$ onto $V_k$. Then $V_k$ is a simple euclidean Jordan algebra of rank $k$ and there is a filtration of euclidean Jordan algebras:
$$
V_1\subset V_2\cdots\subset V_\rho=V.
$$ 

Let ${\bf m}\in\bb Z^\rho$. We write  ${\bf m}=(m_1, \ldots, m_\rho)$ and say that $\bf m\ge 0$ if $m_1\ge\ldots\ge m_\rho\ge 0$. Let 
$$\Delta_{\bf m}(x)=\prod_{i=1}^\rho\Delta_i(x)^{m_i-m_{i+1}},$$
here $m_{\rho+1}=0$, $\Delta_i(x)$ is the determinant of $P_i(x)$, considered as an element of $V_i$.

For ${\bf m}\ge 0$, we let ${\mathcal P}_{\bf m}(V)$\label{Pm} be the subspace of $\mathcal P(V)$ generated by the polynomials $g\cdot \Delta_{\bf m}$, $g\in \mr{Str}(V)$. The polynomials belonging to ${\mathcal P}_{\bf m}(V)$ are homogeneous of degree $|{\bf m}|=\sum m_i$, hence $\mathcal P_{\bf m}(V)$ is finite dimensional. 

\begin{thm}The subspaces $\mathcal P_{\bf m}(V)$ are mutually inequivalent irreducible as representation spaces of $\mr{Str}(V)$, and $\mathcal P(V)$ is the direct sum
$$
{\mathcal P}(V)=\bigoplus_{\bf m\ge 0}{\mathcal P}_{\bf m}(V).
$$
\end{thm}
Since ${\mathcal P}(V)$ consists of complex-valued polynomials on $V$, the representations in this theorem  naturally extends to the complxification of $\mr{Str}(V)$. Using the identification in Eq. (\ref{id: str=k}), these representations naturally becomes
representations of $\frk{u}$. 

From here on, as representations of $ \frk u$, ${\mathcal P}(V)$ and ${\mathcal P}_{\bf m}(V)$ \label{idrep} shall always be viewed in this sense. For later use, we use $\xi_\nu$ to denote the one-dimensional representation $\bb C$ of $ \frk u$ such that
$-{i\over 2}(X_e+Y_e)$ acts as the scalar multiplication by $-\nu{\rho\over 2}$.

\section{Realizations Of TKK Algebras}\label{Realizations of TKK algebras}
The goal of this section is to realize the TKK algebras. The results and their presentations here are strongly influenced by the thinking/practice in physics.  Although our perspective is different, we don't claim any originality here, because most (maybe all) of the materials presented here should be known to the experts in one area or another area. 

We start with the classical realization on symplectic space $T^*V$, from which the operator realizations follow
via the straightforward canonical quantization. Due to the operator ordering ambiguity, we get a family of operator realizations, parametrized by a real parameter $\nu$. The case $\nu={2n/ \rho}$ is well-known to physicists, cf. Ref. \cite{MG93}. The case $\nu > 1+(\rho-1)\delta$ has been worked out by M. Aristidou, M. Davidson and G. \'{O}lafsson \cite{ADO06} by an indirect method.

\subsection{The classical realization of TKK algebras}
As is well-known, the total cotangent space $T^*V$ is a natural symplectic
space. By virtue of the euclidean metric $ds^2_E$ on $V$, one can identify
$T^*V$ with the total tangent space $TV$.  Now the tangent bundle and cotangent
bundle of $V$ both have a natural trivialization, with respect to which, one can
denote an element of $T^*V$ by $(x,p)$ and its corresponding element in $TV$ by
$(x,\pi)$.  We fix an orthonormal basis $\{e_\alpha\}$ for $V$ so that we can write
$x=x^\alpha e_\alpha$ and $\pi=\pi^\alpha e_\alpha$. Then the basic Poisson bracket
relations on $TV$ are $\{x^\alpha, \pi^\beta\}=\delta^{\alpha\beta}$, $\{x^\alpha, x^\beta\}=0$, and $\{\pi^\alpha, \pi^\beta\}=0$.

Introduce the moment functions
\begin{eqnarray}\label{MomentMap}
{\mathcal S}_{uv} :=\langle S_{uv}(x)\mid \pi\rangle, \quad {\mathcal X}_u: =
\langle x\mid \{\pi u\pi \}\rangle, \quad {\mathcal Y}_v: =\langle x\mid
v\rangle
\end{eqnarray} on $TV$. The following theorem would be well-known to experts on
Jordan algebra.
\begin{Thm}\label{Classical Hidden Symmetry}
As polynomial functions on $TV$, ${\mathcal S}_{uv}$, ${\mathcal X}_u$ and
${\mathcal Y}_v$ satisfy
the following Poisson bracket relations: for any $u$, $v$, $z$ and $w$ in $V$, 
\begin{eqnarray}\label{PoissonR}
\left\{\begin{matrix}
\{\mathcal X_u, \mathcal X_v\}=0, \quad \{\mathcal Y_u, \mathcal Y_v\}=0, \quad
\{\mathcal X_u,
\mathcal Y_v\} = -2\mathcal S_{uv},\cr\\ \{\mathcal S_{uv},
\mathcal X_z\}=\mathcal X_{\{uvz\}}, \quad \{\mathcal S_{uv}, \mathcal
Y_z\}=-\mathcal Y_{\{vuz\}},\cr\\
\{\mathcal S_{uv}, \mathcal S_{zw}\} = \mathcal S_{\{uvz\}w}-\mathcal
S_{z\{vuw\}}.
\end{matrix}\right.
\end{eqnarray}
\end{Thm}
\begin{proof} It is clear that $ \{\mathcal Y_u, \mathcal Y_v\}=0$. 
\begin{eqnarray}
\{\mathcal X_u, \mathcal Y_v\} &=&\{\langle x\mid \{\pi u\pi \}\rangle, \langle
x\mid v\rangle \}\cr
&=& -2 \langle x\mid \{v u\pi \}\rangle =-2 \langle S_{uv}(x)\mid \pi\rangle\cr
&=& -2\mathcal S_{uv}.\nonumber
\end{eqnarray}
\begin{eqnarray}
\{\mathcal S_{uv}, \mathcal Y_z\} &=&\{\langle S_{uv}(x)\mid \pi\rangle, \langle
x\mid z\rangle \}\cr
&=&-\langle S_{uv}(x)\mid z\rangle=-\langle x\mid \{vuz\}\rangle\cr
&=& -\mathcal Y_{\{vuz\}}.\nonumber
 \end{eqnarray}
 \begin{eqnarray}
\{\mathcal S_{uv}, \mathcal S_{zw}\} &=&\{\langle S_{uv}(x)\mid \pi\rangle,
\langle S_{zw}(x)\mid \pi\rangle \}\cr
&=&\langle S_{uv}S_{zw}(x)\mid \pi\rangle -\langle S_{zw}S_{uv}(x)\mid \pi\rangle\cr
&=& \langle [S_{uv},S_{zw}](x)\mid \pi\rangle=\langle (S_{\{uvz\}w}-
S_{z\{vuw\}})(x)\mid \pi\rangle\cr
&=&\mathcal S_{\{uvz\}w} - \mathcal S_{z\{vuw\}}.\nonumber
 \end{eqnarray}
\begin{eqnarray}
\{\mathcal S_{uv}, \mathcal X_z\} &=&\{\langle S_{uv}(x)\mid \pi\rangle, \langle
x\mid \{\pi z\pi\}\rangle \}\cr
&=&-\langle S_{uv}(x)\mid \{\pi z\pi\}\rangle + 2\langle x\mid \{\pi
z\{vu\pi\}\}\rangle\cr
&=& \langle x\mid 2S_{\pi z}S_{vu}(\pi)-S_{vu}S_{\pi z}(\pi)\}\rangle\cr
&= &  \langle x\mid S_{\pi z}S_{vu}(\pi)-[S_{vu},S_{\pi z}](\pi)\}\rangle\cr
&=& \langle x\mid S_{\pi z}(\{vu\pi\})-S_{\{vu\pi\}
z}(\pi)+S_{\pi\{uvz\}}(\pi)\}\rangle\cr
&=& \langle x\mid \{\pi\{uvz\}\pi\}\rangle\cr
&=& \mathcal X_{\{uvz\}}.\nonumber
\end{eqnarray}
Finally, 
\begin{eqnarray}
\{\mathcal X_u, \mathcal X_v\} &=&\{\langle x\mid \{\pi u\pi \}\rangle, \langle
x\mid \{\pi v \pi\rangle \}\cr
&=& 2 \langle x\mid \{\pi v \{\pi u\pi\}\}\rangle- 2\langle x\mid \{\pi u
\{\pi v\pi\}\}\rangle\cr
&=& 2\langle x\mid [S_{\pi v},S_{\pi u}]\pi\rangle\\
&=& \langle x\mid [S_{\pi v},S_{\pi u}]\pi\rangle-\langle x\mid [S_{\pi u},S_{\pi v}]\pi\rangle\cr
&=&  \langle x\mid S_{\{\pi v\pi\}u}\pi-S_{\pi\{v\pi u\}}\pi\rangle- \langle x\mid S_{\{\pi u\pi\}v}\pi-S_{\pi\{u\pi v\}}\pi\rangle\cr
&=&  \langle x\mid S_{\{\pi v\pi\}u}\pi\rangle- \langle x\mid S_{\{\pi u\pi\}v}\pi\rangle\quad \mbox{because  $\{u\pi v\}=\{v\pi u\}$}\cr&=& - \langle x\mid [S_{\pi v},S_{\pi u}]\pi\rangle\\
&=&0,\nonumber
\end{eqnarray}
because it is equal to the negative half of itself.
\end{proof}

\subsection{The operator realizations of TKK Algebras }
The canonical quantization involves promoting classical physical variables $\mathcal O$ to
differential operators $\hat  O$ (or the duals
$\check O$) using recipe: $\pi_\alpha\to -i{\partial \over \partial
x^\alpha}$ (or $x_\alpha\to i{\partial \over \partial \pi^\alpha}$). Here is a
word of warning:  in order to get anti-hermitian differential operators in the end, instead of
using the quantized differential operators, we actually use the quantized differential operators multiplied
by $-i$.

For simplicity we write $\sum_\alpha e_\alpha {\partial \over \partial x^\alpha}$ as
${/\hskip -5pt \partial}$ and $\sum_\alpha e_\alpha {\partial \over \partial \pi^\alpha}$ as
${\backslash    
 \hskip -6pt \partial}$. We introduce differential operators on $V$:
\begin{eqnarray}
{\hat S}_{uv} :=-\langle S_{uv}(x)\mid {/\hskip -5pt \partial}\rangle,\quad
{\hat X}_u : = i \langle x\mid \{{/\hskip -5pt \partial} u {/\hskip -5pt
\partial} \}\rangle, \quad
{\hat Y}_v := -i\langle x\mid v\rangle.
\end{eqnarray}
and differential operators on $V^*$:
\begin{eqnarray}
{\check S}_{uv} :=\langle S_{vu}(\pi)\mid {\backslash    
 \hskip -6pt \partial}\rangle,\quad
{\check X}_u : =  \langle \{\pi u \pi\}\mid {\backslash    
 \hskip -6pt \partial}\rangle, \quad
{\check Y}_v := \langle v\mid {\backslash    
 \hskip -6pt \partial}\rangle.
\end{eqnarray}
It is easy to see that the TKK commutation relations (\ref{TKKRel}) hold when all $O$ there are replaced by either their hat version or their check version. Note that the check version is well-known to physicists, cf. Ref. \cite{MG93}. However, the quantization has ambiguity because of the operator ordering problem.
To get the general version of quantization, we let $\nu$ be a real parameter and
introduce differential operators on $V$:
\begin{eqnarray}\label{OperatorR}
{\acute S}_{uv}(\nu) :=\hat S_{uv}-{\nu\over 2}\tr (uv),\quad
{\acute X}_u (\nu): = \hat X_u+i\nu\tr(u{/\hskip -5pt \partial}), \quad
{\acute Y}_v(\nu) :=\hat Y_v. 
\end{eqnarray}
 and differential operators on $V^*$:
 \begin{eqnarray}
{\grave S}_{uv}(\nu) :=\check S_{uv}-{\nu^*\over 2}\tr (uv),\quad
{\grave X}_u (\nu): = \check X_u-\nu^*\tr(u\pi), \quad
{\grave Y}_v(\nu) :=\check Y_v.\end{eqnarray}
where $\nu^*=\nu-{2n\over \rho}$.
\begin{Thm}\label{universalQ} 
 The TKK commutation relations (\ref{TKKRel}) still hold when all $O$
there are replaced by either their acute version or their grave version. 
\end{Thm}
\begin{proof} 
When $\nu=0$, the proof is essentially the same as the proof of Theorem
\ref{Classical Hidden Symmetry}, so we skip it. For the general case, we shall verify the acute
version and leave the grave version to the readers.

Verify that $[\acute Y_u(\nu), \acute Y_v(\nu)]=0$: 
$$[\acute Y_u(\nu), \acute Y_v(\nu)]= [\hat Y_u, \hat Y_v]=0.$$

Verify that $[\acute S_{uv}(\nu),
\acute Y_z(\nu)]=-\acute Y_{\{vuz\}}(\nu)$: 
$$[\acute S_{uv}(\nu), \acute Y_z(\nu)] = [\hat S_{uv},\hat Y_z]=-\hat
Y_{\{vuz\}}=-\acute Y_{\{vuz\}}(\nu).$$

Verify that $[\acute S_{uv}, \acute S_{zw}]=\acute S_{\{uvz\}w}-\acute
S_{z\{vuw\}}$: 
Since $[\hat S_{uv}, \hat
S_{zw}]=\hat S_{\{uvz\}w}-\hat S_{z\{vuw\}}$, all we need to
check is that
$$
\tr(\{uvz\}w)=\tr(z\{vuw\}),\quad\mbox{i.e.,}\quad \langle S_{uv}(z)\mid w
\rangle= \langle z\mid S_{vu}(w) \rangle
$$ which is true because $S_{uv}'=S_{vu}$.

Verify that $[\acute X_u(\nu), \acute
Y_v]=-2\acute S_{uv}$: 
\begin{eqnarray}
[\acute X_u(\nu), \acute
Y_v(\nu)] &=&[\hat X_u+i\nu\tr(u{/\hskip -5pt \partial}), \hat
Y_v]\cr 
&=& -2\hat S_{uv}+\nu\tr(uv)\cr 
&= &-2\acute
S_{uv}(\nu).\nonumber
\end{eqnarray}

Verify that $[\acute S_{uv}(\nu), \acute
X_z(\nu)]=\acute X_{\{uvz\}}(\nu)$:
\begin{eqnarray}
[\acute S_{uv}(\nu), \acute
X_z(\nu)] &=&[\hat S_{uv}(\nu), \hat X_z+i\nu\tr(z{/\hskip -5pt \partial})]\cr 
&=& \hat X_{\{uvz\}}+i\rho\nu[\hat S_{uv}, \langle z\mid {/\hskip -5pt
\partial}\rangle]\cr
&=& \hat X_{\{uvz\}}+i\rho\nu\langle S_{uv}(z)\mid {/\hskip -5pt
\partial}\rangle\cr
&=&  \acute X_{\{uvz\}}(\nu) .\nonumber
\end{eqnarray}

Verify that $[\acute X_u(\nu), \acute X_v(\nu)]=0$: 
\begin{eqnarray}
[\acute X_u(\nu), \acute X_v(\nu)]&=& [\hat X_u+i\rho\nu\langle u\mid {/\hskip
-5pt \partial}\rangle, \hat X_v++i\rho\nu\langle v\mid {/\hskip -5pt
\partial}\rangle]\cr
&=&i\rho\nu[\hat X_u, \langle v\mid {/\hskip -5pt \partial}\rangle] - <u\leftrightarrow v>\cr
&=&-\rho\nu[\langle x\mid \{{/\hskip -5pt \partial}u{/\hskip -5pt
\partial}\}\rangle, \langle v\mid {/\hskip -5pt \partial}\rangle] -
<u\leftrightarrow v>\cr
&=&\rho\nu\langle v\mid \{{/\hskip -5pt \partial}u{/\hskip -5pt
\partial}\}\rangle - <u\leftrightarrow v>\cr
&=&\rho\nu \langle \{u{/\hskip -5pt \partial} v\}\mid {/\hskip -5pt
\partial}\rangle -
<u\leftrightarrow v>\cr
&=&0.\nonumber
\end{eqnarray}
Here,  $<u\leftrightarrow v>$ means the term same as the one on the left except that $u$ and $v$ are interchanged.

\end{proof}
In the remainder of this paper, let us focus the attention on the operator realization on $V$:
$O\to\acute O(\nu)$ for a fixed $\nu$. Note that this operator realization provides a linear action of the TKK algebra on $C^\infty(V)$. We shall investigate the unitarity of this action in the next section. In order to do that, let us make some preparations here.

Let $\mathcal P(V)$ be the algebra of $\bb C$-valued polynomial functions on $V$, and $\mathcal P_I(V)$  be the vector subspace consisting of polynomials of degree at most $I$. Let
\begin{eqnarray}\acute D(V)= e^{-r}\mathcal P(V), \quad \acute D_I(V)= e^{-r}\mathcal P_I(V).
\end{eqnarray}
It is clear that the action of $\frk{co}$ on $C^\infty(V)$, which maps $O$ to $\acute O(\nu)$,  leaves $\acute D(V)$ invariant. Let us denote by $\pi_\nu$ this action on $\acute D(V)$.

Recall that $\nu$ is a real parameter and $\frk{u}$ is the maximal compact Lie subalgebra of $\frk{co}$.
\begin{Thm} \label{universalQ2} Let $H_e = i( X_e+ Y_e)$, $\pi_\nu|_{ \frk u}$ be the restriction of $\pi_\nu$ from $\frk{co}$ to $\frk{u}$.

i) $\pi_\nu|_{ \frk u}$ leaves $\acute D_I(V)$ invariant and commutes with the inclusion of $\acute D_{I-1}(V)$ into  $\acute D_I(V)$, consequently it induces a linear action on ${\acute D}_I(V)/{\acute D}_{I-1}(V)$. 

ii) As a representation of $\frk{u}$, we have isomorphism
\begin{eqnarray}
\begin{array}{ccc}
\xi_\nu\otimes\bigoplus_{{\bf m}\ge 0}^{|{\bf m}|=I} {\mathcal P}_{ {\bf m}}(V) &\to & {\acute D}_I(V)/{\acute D}_{I-1}(V)\\
\\
1\otimes p &\mapsto &e^{-r}p+{\acute D}_{I-1}(V).\nonumber
\end{array}
\end{eqnarray}

iii) The induced linear map $$\acute H_e:  \quad{\acute D}_I(V)/{\acute D}_{I-1}(V)\to {\acute D}_I(V)/{\acute D}_{I-1}(V)$$
is the scalar multiplication by $(2I+\nu\rho)$.

iv) $\pi_\nu$ is unitarizable $\implies$ $\nu \ge 0$.

v) $e^{-r}$ is a lowest weight state with
weight $\nu\lambda_0$.

vi) $\pi_\nu$ is indecomposable.

\end{Thm}
\begin{proof}
i) Since
\begin{eqnarray}\label{identificationR}
e^r {-i\over 2}(\acute X_u+\acute Y_u) e^{-r}={1\over 2}(\langle x\mid \{{/\hskip -5pt \partial}u{/\hskip -5pt \partial}\}\rangle+\nu\tr (u{/\hskip -5pt \partial}))+\hat
L_u-{\nu\over 2}\tr u,
\end{eqnarray}
$(\acute X_u+\acute Y_u)$ maps $\acute D_I$ into $\acute D_I$. It is also clear that $[\acute L_u, \acute L_v]$ maps $\acute D_I$
into $\acute D_I$. Therefore, in view of the fact that $$\frk{u}=\mr{span}_{\bb R}\{[L_u, L_v], X_w+Y_w\mid u, v, w\in
V\},$$
the linear action of $\frk{u}$
on $\acute D$ leaves $\acute D_I$ invariant. The rest is clear.

ii) That is clear from Eq. (\ref{identificationR}) and the last paragraph of Subsection \ref{SS: decomposition}. 

iii) For any homogeneous degree $I$ polynomial $p$, we have
\begin{eqnarray}
i(\acute X_e+\acute Y_e) e^{-r} p &\equiv &e^{-r} (-2\hat L_e+\nu\rho)p \mod
\acute D_{I-1}\cr
&\equiv & (2I+\nu\rho)e^{-r}p \mod \acute D_{I-1}.\nonumber
\end{eqnarray}

iv) Let $E_{\pm} : = {i\over 2}( X_e- Y_e)\mp S_{ee}$. Then 
$$
 [H_e, E_\pm]=\pm2 E_\pm,\quad [E_+, E_-]=-H_e.
$$ 
Suppose that $\pi_\nu$ is unitarizable, and $(,)$ is the inner product on
$\acute D$. Let \fbox{$\psi_0=e^{-r}$}.  Since $||\pi_\nu(E_+)\psi_0 ||^2\ge 0$ and
$\pi_\nu(E_-)\psi_0=0$, using relation
$[E_+, E_-]=-H_e$, we arrive at $(\psi_0, \pi_\nu(H_e)\psi_0)\ge 0$, i.e., 
$\nu\rho||\psi_0||^2\ge 0$. So $\nu\ge 0$.

v) Let us take the simple root system $\alpha_0$, \ldots, $\alpha_r$
specified in Lemma \ref{LemmaVogan}. Since $\acute D_0$ ($=\mr{span}_{\bb C}\{\psi_0\}$) is one dimensional
and $\bar{ \frk u}$ is semi-simple, the action of  $\bar{ \frk u}^{\bb C}$ on $\acute D_0$ must be
trivial. Therefore, for $i\ge 1$,  in view of the fact that $E_{\pm\alpha_i}, H_{\alpha_i}\in
\bar{ \frk u}^{\bb C}$, we have
\begin{eqnarray}\label{lw1}
\pi_\nu(E_{-\alpha_i}) \psi_0=0, \quad \pi_\nu(H_{\alpha_i}) \psi_0=0.
\end{eqnarray}
On the other hand, since $E_{-\alpha_0}={i\over
2}(X_{e_{11}}-Y_{e_{11}})+L_{e_{11}}$ and
$H_{\alpha_0}=i(X_{e_{11}}+Y_{e_{11}})$, by a computation, we have
\begin{eqnarray}
e^r\acute E_{-\alpha_0}e^{-r} &=& -{1\over 2}\left(\langle x\mid \{{/\hskip -5pt
\partial}e_{11}{/\hskip -5pt \partial}\}\rangle+\nu\tr (e_{11}{/\hskip -5pt
\partial})\right),\cr
e^r\acute H_{\alpha_0}e^{-r} &=& -\langle x\mid \{{/\hskip -5pt
\partial}e_{11}{/\hskip -5pt \partial}\}\rangle-\nu\tr (e_{11}{/\hskip -5pt
\partial})-2\hat L_{e_{11}}+\nu,\nonumber
\end{eqnarray}
so it is easy to see that
\begin{eqnarray}\label{lw2}
\pi_\nu(E_{-\alpha_0})\psi_0=0, \quad \pi_\nu(H_{\alpha_0})\psi_0=\nu\psi_0.
\end{eqnarray}
Therefore,  in view of the fact that $\alpha_0(H_{\alpha_0})=2$, Eqns
(\ref{lw1}) and (\ref{lw2})
imply that $\psi_0$ is a lowest weight state with weight $\nu\lambda_0$.

vi) In view of the fact that operator $\acute Y_v$ is the multiplication by $-i\langle v\mid x\rangle$, this is obvious: 
$$
e^{-r}\sum_{i_1, \ldots, i_n} \alpha_{i_1\cdots i_n} x_1^{i_1}\cdots x_n^{i_n}=\left(\sum_{i_1, \ldots, i_n} \alpha_{i_1\cdots i_n} (i\acute Y_{e_1})^{i_1}\cdots(i\acute Y_{e_n})^{i_n}\right) \psi_0.
$$

\end{proof}

We shall show in the next section that $\pi_\nu$ is irreducible when $\nu>(\rho-1){\delta\over 2}$ and is not irreducible when $\nu=k{\delta\over 2}$, $k=0$, $1$, \ldots, $(\rho-1)$. The collection of theses values of $\nu$, denoted by $\mathcal W(V)$, is called the {\bf Wallach set} for $V$. So
$$
{\mathcal W}(V)=\left\{k{\delta\over 2}\mid k=0, 1, \ldots, \rho-1\right\}\cup\left((\rho-1){\delta\over 2}, \infty\right).
$$
It is known from Ref. \cite{EHW82} that the set of scalar-type unitary lowest weight representation of $\mr{Co}$ is isomorphic to ${\mathcal W}(V)$. We shall show in the next section that these representations are precisely the irreducible quotient of these $\pi_\nu$, also denoted by $\pi_\nu$. 

As mentioned in the introduction, the operator realizations as given in Eq. (\ref{OperatorR}) are not unitary with respect to the obvious $L^2$-inner product
$$
(\psi_1, \psi_2)=\int_V \bar \psi_1\psi_2\, dm,
$$
where $dm$ is the Lebesgue measure. The right inner product  must be found in order to have unitary realizations. For that, let us move on to the next section.

\section{Quantizations Of TKK Algebras}\label{quantization}
The goal of this section is to investigate the unitarity of representation $\pi_\nu$ obtained in the previous section. The question is to find a positive hermitian form $(, )_\nu$ on $\acute D(V)$ with respect to which operators $\acute O(\nu)$ are all anti-hermitian. More generally,   $(, )_\nu$ can be semi-positive because then $\pi_\nu$ descends to a unitary representation by formally setting the ``spurious states" (i.e., elements of $\acute D(V)$ with zero norm) as zero. It is a fact from Ref. \cite{FK94} that such a $(, )_\nu$ does exist for $\nu$ in the Wallach set. Our purpose here is to present a new route towards this fact along with its refinements. 

The case $\nu={\delta\over 2}$ is already known from our work in Ref. \cite{meng09}. Recall from Ref. \cite{meng09}, the $J$-Kepler problem for $V$ is a dynamic problem on $\ms P$ --- the submanifold of $V$ consisting of semi-positive elements of rank one. 
By comparing Remark 8.1 and Proposition 8.1 from Ref. \cite{meng09} with Theorem \ref{universalQ2} here,  we have
\begin{eqnarray}\label{hermitionform}
(\psi_1, \psi_2)_{\delta\over 2}=\int_{\ms P}\bar\psi_1\, \psi_2\, r^{-1-(\rho/2-1)\delta}\,\mr{vol}.
\end{eqnarray}
Here, $\mr{vol}$ is the volume form for the {\bf Kepler metric}
\begin{eqnarray}\label{KeplerM}
ds^2_K:={2\over \rho}ds^2_E-dr^2.
\end{eqnarray}
It is clear that, as a vector subspace of $\acute D(V)$, the space of ``spurious states" consists
of elements of $\acute D(V)$ which vanish on $\ms P$.

With this result in mind, it is not hard to imagine what the general picture should be: replacing $\ms P$ by the submanifold of $V$ consisting of semi-positive elements of a fixed positive rank. Of course, some technical hurdles must be overcome.
The initial hurdle is the generalization of the Kepler metric in Eq. (\ref{KeplerM}). The second hurdle is the generalization of the extra factor $r^{-1-(\rho/2-1)\delta}$ in measure $$d\mu_{\delta\over 2}:=r^{-1-(\rho/2-1)\delta}\,\mr{vol}.$$

It turns out, the second hurdle simply disappears by itself as we walk along a natural path towards quantizations of TKK algebras. The clue for removing the first hurdle comes from the study of the universal Kepler problem in Ref. \cite{meng10}, as we shall sketch below. 

We have noted in the past that the total tangent space of a {\it Riemannian} manifold is a symplectic manifold, and if $N$ is a submanifold of $M$, then $TN$ is a symplectic submanifold of $TM$. With this understood, we remarked in Ref. \cite{meng10}
that, by restricting the classical universal hamiltonian 
$$
{\mathcal H}={1\over 2}{\langle x\mid \pi^2\rangle\over r}-{1\over r}
$$ 
from $TV$ to $T\ms P$, one obtains the classical hamiltonian for the J-Kepler problem. Since the first term in $\mathcal H$ should be identified with the kinetic energy, we must have the following new formula for the Kepler metric \footnote{Here, for endomorphism $A$ on $V$ and elements $u$, $v$ in $V$, Dirac's notation $\langle u\mid A  \mid v\rangle$ means $\langle u\mid A v\rangle$.}:
\begin{eqnarray}\label{KineticE}
(\pi, \pi)_{ds^2_K}={\langle x\mid \pi^2\rangle\over r}={\langle \pi\mid L_x \mid \pi\rangle\over r},
\end{eqnarray}
a fact which can be verified directly. Now it becomes clear how to generalize the Kepler metric.

\subsection{Canonical cones}\label{SS:Canonical Cones}
We say that an element $x\in V$ is {\it semi-positive} if $x=y^2$ for some $y\in
V$.
Let us denote by ${\mathcal Q}$ the space of semi-positive elements in $V$ and recall that $\mr {Str}$ is
the structure group
of $V$. Then one can check that the action of $\mr{Str}$  on $V$ leaves ${\mathcal Q}$ invariant, so
we have a partition of ${\mathcal Q}$ into the disjoint union of
$\mr{Str}$-orbits:  $${\mathcal Q}=\cup_{k=0}^\rho {\mathcal C}_k. $$ Here, 
homogeneous space ${\mathcal C}_k$ is the space of semi-positive elements of
rank $k$. 
Note that,  ${\mathcal C}_0=\{0\}$,  ${\mathcal C}_\rho$ is the symmetric cone
$\Omega$ of $V$, and $\mathcal Q$ is the topological closure of $\Omega$. 

As a sub-manifold of the Euclidean space $V$,  ${\mathcal C}_k$  has an induced
Riemannian metric. When we say that $T{\mathcal C}_k$ is a symplectic manifold, it is
this Riemannian metric that is used to identify $T{\mathcal C}_k$ with $T^*{\mathcal C}_k$. However, the Riemannian metric for the Kepler-type dynamics on ${\mathcal C}_k$ is a different one, which we shall describe below.

For any $u\in V$, since  $L_u$:  $V\to V$ is self-adjoint, we have an orthogonal
decomposition $V=\mr{Im} L_u\oplus \ker L_u $, with respect to which, $L_u$
decomposes as $L_u=\bar L_u\oplus 0$. Since $\bar L_u$ is invertible, we can
introduce endomorphism
\begin{eqnarray}
{1\over L_u}\buildrel\mr{def}\over =  {\bar L_u}^{-1}\oplus 0
\end{eqnarray} on $V$.  For any $x\in {\mathcal C}_k$, one can check that the
tangent space $T_x{\mathcal C}_k$ is $\{x\}\times \mr{Im} L_x$, the normal space $N_x{\mathcal C}_k$ is
$\{x\}\times \ker L_x$ .
\begin{Definition}[Canonical Metric]\label{canonical metric}
The {\bf canonical metric} on ${\mathcal C}_k$, denoted by $ds^2_K$, is defined as follows:
\begin{eqnarray}
T_x{\mathcal C}_k\times T_x{\mathcal C}_k & \to &  \bb R\cr
((x,u), (x, v)) &\mapsto &r  \langle u\mid {1\over L_x}\mid v\rangle =r  \langle
u\mid {\bar L_x}^{-1}\mid v\rangle.
\end{eqnarray} 
\end{Definition}
One can check that, on ${\mathcal C}_1$, the canonical metric is the Kepler
metric introduced in Ref. \cite{meng09}. On the symmetric cone $\Omega$, the
canonical metric is
\begin{eqnarray}
T_x\Omega\times T_x\Omega & \to &  \bb R\cr
((x,u), (x, v)) &\mapsto &r  \langle u\mid {L_x}^{-1}\mid v\rangle.
\end{eqnarray} 

\begin{Definition}[Canonical Cone]\label{D:main1} 
Let $V$ be a simple euclidean Jordan algebra of rank $\rho$, ${\mathcal C}_k$
be its submanifold  consisting of the semi-positive elements of rank $k$, $1\le k\le
\rho$. The $V$'s {\bf canonical cone} of rank $k$ is defined to be the Riemannian manifold $({\mathcal
C}_k, ds^2_K)$, where $ds^2_K$ is the canonical metric in Definition \ref{canonical
metric} .
\end{Definition}
We remark that, since the action of structure group on $V$ leaves
 ${\mathcal C}_k$ invariant, for any $u, v\in V$ and $x\in {\mathcal C}_k$, $\hat
S_{uv}|_x\in T_x{\mathcal C}_k$; i.e., $\hat S_{uv}$ descends to a vector field
on ${\mathcal C}_k$, which shall still be denoted by $\hat S_{uv}$. Since $L_u=S_{ue}$, we write
$\hat L_u$ for $\hat S_{ue}$. In the remainder of this paper, we shall use $\Delta$ ($\mr{vol}$ resp.) denote the Laplace operator (the volume form resp.) on a Riemannian manifold, e.g., a canonical cone.

\subsection{Basic facts on canonical cones}
We use $Tr$ to denote the trace for endomorphisms on $V$. For any $x$ in the
canonical cone, we use $P_x$: $V\to V$ to denote the orthogonal projection onto
$\mbox{Im }L_x$. Throughout this subsection,  we focus our attention on a fixed
canonical cone $\mathcal C$ of rank $k$.

Let us start with a local analysis of the canonical metric around a point $x_0\in
{\mathcal C}$. Choose a Jordan frame
$\{e_{ii}\}_{i=1}^\rho$ such that
$$
x_0=\sum_{i=1}^k\lambda_i e_{ii}
$$ for some positive numbers $\lambda_1$, \ldots, $\lambda_k$. With this Jordan frame $\{e_{ii}\}_{i=1}^\rho$ fixed, we let $N:= \bigoplus_{j\ge
i\ge k} V_{ij}$  and $T$ be the orthogonal complement of $N$ in (euclidean
vector space) $V$. Then $T_{x_0}\mathcal C=\{x_0\}\times T$. 

For any $x\in V$, one can
uniquely decompose $x=x_0+t+y$ with $t\in T$ and $y\in N$. Now if we assume that
$x$ is always in ${\mathcal C}$, then 
$$
y=O(|t|^2)\quad\mbox{near $x_0$. }
$$
Choose basis $\{e_i\}$ for $T$ which is  orthonormal with respect to the inner
product on $V$. Write $t=t^i e_i$ and the
canonical metric $ds^2_K$ as  $h_{ij} dt^i\,dt^j$, then
$$
h_{ij}(x) =\langle e\mid x\rangle  \langle x_i\mid {\bar L_{x}}^{-1}\mid x_j\rangle
,$$
where $x_i={\partial x\over \partial t^i}$. Let $g_{ij}:=\langle x_i\mid
x_j\rangle $,
then
$$
g_{ij}(x)=\delta_{ij}+O(|t|^2)\quad\mbox{ near $x_0$.}
$$

As usual, the inverse of $[g_{ij}]$ is denoted by $[g^{ij}]$ and  the inverse of
$[h_{ij}]$ is denoted by $[h^{ij}]$. It is easy to see that
$$
h^{ij}(x)={g^{ik}(x) \langle x_k\mid L_{x}\mid x_l\rangle g^{lj}(x)\over \langle e\mid
x\rangle }. 
$$ 
Consequently, when $x\in {\mathcal C}$ is near $x_0$, we have
\begin{eqnarray}\label{h-estimate}
h^{ij} (x)&= &{ \langle x_i\mid L_{x}\mid x_j\rangle  \over \langle e\mid x\rangle
}+O(|t|^2) \cr
&=& { \langle e_i\mid x_0 e_j\rangle  \over \langle e\mid x_0\rangle
}\left(1-{\langle e\mid t\rangle \over \langle e\mid x_0\rangle }\right)+{
\langle e_i\mid t e_j\rangle  \over \langle e\mid x_0\rangle }+O(|t|^2).
\end{eqnarray}

\begin{Prop}  \label{PropLu1} Fix a canonical cone. Let $\mr{vol}$ be its volume form, and $\ms L_u$
the Lie derivative with respect to vector field $\hat L_u$.  Then
\begin{eqnarray}\label{lambda-uID}
\ms L_u({1\over r}\mr{vol})=-2\lambda_u{1\over r}\mr{vol},
\end{eqnarray}
where
\begin{eqnarray}
\fbox{$2\lambda_u=-{1\over 2}\mr{Tr} ({1\over L_x}L_{ux})+\mr{Tr} (P_xL_u)+{\langle u\mid
x\rangle \over \langle e\mid x \rangle }\left({\mr{Tr} P_x\over 2} -1\right).$}
\end{eqnarray}
Consequently, $\lambda_u$ depends on $u$ linearly and 
\begin{eqnarray}\label{lambdaP}
\hat L_u(\lambda_v)=\hat L_v(\lambda_u), \quad 
{\hat S}_{uv}(\lambda_{zw})-{\hat S}_{zw}(\lambda_{uv})=\lambda_{
\{uvz\}w}-\lambda_{z\{vuw\}} .
\end{eqnarray}
\end{Prop}
\begin{proof}
We just need to show that
\begin{eqnarray}
\ms L_u(\mr{vol})=-2\tilde \lambda_u \mr{vol}\nonumber
\end{eqnarray}
at a point $x_0$, where
\begin{eqnarray}\label{i1}
-2\tilde \lambda_u={1\over 2}\mbox{Tr} \left(-{\langle u\mid x\rangle \over
\langle e\mid x\rangle }P_x+{1\over L_x}L_{xu}\right)-\mr{Tr} (P_xL_u).
\end{eqnarray}

Recall that $ds_K^2=h_{ij}\,dt^i\,dt^j$. Since $\mr{vol}=\sqrt h\, dt^1\wedge dt^2\wedge\cdots$ where $h=\det[h_{ij}]$, we
have
$$
{\ms L}_u(\mr{vol})=\hat L_u(\ln \sqrt h) \mr{vol} +\sum_i \sqrt h\,
dt^1\wedge\cdots\wedge{\ms L}_u(dt^i) \wedge \cdots.
$$
Since ${\ms L}_u(dt^i)=d{\ms L}_u(\langle e_i\mid x\rangle )
= -d (\langle xe_i\mid u\rangle )$, we have
$${\ms L}_u(dt^i) = -\langle e_i^2\mid u\rangle dt^i\quad \mbox{at $x_0$},
$$
consequently,
\begin{eqnarray}\label{i2}
-2\tilde \lambda_u=\hat L_u(\ln \sqrt h) -\langle \sum_i e_i^2\mid u\rangle
\quad \mbox{at $x_0$}.
\end{eqnarray}
In view of Eq. (\ref{h-estimate}), for $x$ near $x_0$,
$$
h=\det [{\langle e_i\mid x_0e_j\rangle \over \langle e\mid x_0\rangle
}]^{-1}\left(1+\mbox{Tr} \left({\langle e\mid t\rangle \over \langle e\mid
x_0\rangle }P_{x_0}-{1\over L_{x_0}}L_t\right)\right)+O(|t|^2),
$$
so
\begin{eqnarray}\label{i3}
\hat L_u(\ln \sqrt h) |_{x_0}={1\over 2}\mbox{Tr} \left(-{\langle u\mid x_0\rangle \over
\langle e\mid x_0\rangle }P_{x_0}+{1\over L_{x_0}}L_{x_0u}\right).
\end{eqnarray}
Since
$$
 {1\over \rho}\sum_i e_i^2=(1+{\delta(\rho-k-1)\over 2})\sum_{j=1}^k
e_{jj}+{\delta k\over 2} e,
$$ 
where $k$ is the rank of the canonical cone, we have
\begin{eqnarray}\label{i4}
 \langle \sum_i e_i^2\mid u\rangle  &= &(1+{\delta(\rho-k-1)\over
2})\sum_{j=1}^k u_{jj}+{\delta k\over 2} \tr u\cr
&=& \mr{Tr}(P_{x_0}L_u). 
\end{eqnarray}
Plugging Eqs. (\ref{i3}) and (\ref{i4}) into Eq. (\ref{i2}), we arrive at Eq.
(\ref{i1}).

Since $[{\ms L}_{{\hat L}_u}, {\ms L}_{{\hat L}_v}] = {\ms L}_{[\hat L_u, \hat
L_v]}$, in view of fact that the Kepler metric is invariant under the action of the
$\mr{Aut}(J)$,  we have $[{\ms L}_{{\hat L}_u}, {\ms L}_{{\hat L}_v}] ({1\over
r}\mr{vol})=0$. Then
$$-2\left({\hat L}_u(\lambda_v) - {\hat L}_v(\lambda_u)\right) \cdot {1\over
r}\mr{vol} =0,$$ 
consequently ${\hat L}_u(\lambda_v) = {\hat L}_v(\lambda_u)$. 

Since $[{\ms L}_{{\hat S}_{uv}}, {\ms L}_{{\hat S}_{zw}}]={\ms L}_{[{\hat
S}_{uv}, {\hat S}_{zw}]}={\ms L}_{{\hat S}_{\{uvz\}w}-{\hat S}_{z\{vuw\}}}$,
acting on ${1\over r}\mr{vol}$, we have 
$$
{\hat S}_{uv}(\lambda_{zw})-{\hat S}_{zw}(\lambda_{uv})=\lambda_{
\{uvz\}w}-\lambda_{z\{vuw\}} .
$$
\end{proof} 
In the remainder of this paper, we let
\begin{eqnarray}
\fbox{${\tilde S}_{uv}={\hat S}_{uv}-\lambda_{uv}$, \quad $\tilde L_u=\tilde S_{ue}$.}
\end{eqnarray}

\begin{Prop}  \label{PropLu2}
Fix a canonical cone and let $\Delta$ be its Laplace operator. Then
\begin{eqnarray}\label{KeyId1'}
[r\Delta, \langle u\mid x\rangle]=-2\tilde L_u, \quad u\in V.
\end{eqnarray}
\end{Prop}
\begin{proof} Upon observing that $[r\Delta, \langle u\mid x\rangle ]$ is a linear
differential operator, it suffices to verify that
$$
[[r\Delta, \langle u\mid x\rangle ], \langle v\mid x\rangle ](1)=[-2\tilde L_u,
\langle v\mid x\rangle ](1),  \quad [r\Delta, \langle u\mid x\rangle ]
(1)=-2\tilde L_u (1), 
$$
i.e.,
\begin{eqnarray}\label{verification}
[[r\Delta, \langle u\mid x\rangle ], \langle v\mid x\rangle ](1)=2 \langle
uv\mid x\rangle ,  \quad r\Delta (\langle u\mid x\rangle )= 2\lambda_u. 
\end{eqnarray}
In view of the fact that  $*\Delta f =  d*df$,  $[[*\Delta, \langle u\mid
x\rangle ], \langle v\mid x\rangle ] (1)$ is equal to
\begin{eqnarray}
&& d*d(\langle u\mid x\rangle \langle v\mid x\rangle ) -\langle u\mid x\rangle
d*d(\langle v\mid x\rangle )-\langle v\mid x\rangle d*d(\langle u\mid x\rangle
)\cr
&=& d(\langle u\mid x\rangle )\wedge *d(\langle v\mid x\rangle )+d\langle v\mid
x\rangle \wedge *d(\langle u\mid x\rangle )\cr
&=& 2\langle u\mid x_i\rangle h^{ij} \langle v\mid x_j\rangle \;
\mr{vol},\nonumber
\end{eqnarray} 
so 
$$
[[r\Delta, \langle u\mid x\rangle ], \langle v\mid x\rangle ] (1)|_{x_0} =
2\langle u\mid e_i\rangle \langle e_i\mid L_{x_0}|e_j\rangle  \langle v\mid
e_j\rangle =2\langle uv\mid x_0\rangle .
$$
The first identity of Eq. (\ref{verification}) is verified.

In view of the fact that $*\Delta(f)=d\left(\sum_{i,j}h^{ij} \partial_i f
\iota_{\partial_j} (\mr{vol})\right)$, we have
\begin{eqnarray}
*r\Delta(\langle u\mid x\rangle )|_{x_0} &= & \left.r d\left(\sum_{i,j}h^{ij}
\langle u\mid x_i\rangle  \iota_{\partial_j} (\mr{vol})\right)\right |_{x_0}\cr
&=&- \left.r d\left({1\over r} \iota_{\hat L_u} (\mr{vol})\right)\right
|_{x_0}\cr
&=& -{\ms L}_u(\mr{vol})|_{x_0}-{\langle e\mid ux\rangle \over
r}\mr{vol}|{x_0}\cr
&=& -r{\ms L}_u({1\over r} \mr{vol})|_{x_0}\cr
&\buildrel {Eq. (\ref{lambda-uID})}\over =&2\lambda_u\mr{vol}|_{x_0}.\nonumber
\end{eqnarray}
The second identity of Eq. (\ref{verification}) is verified.
\end{proof}

As we have demonstrated in Ref. \cite{meng09}, to check a commutation relation
on the Kepler cone, it is easier to check the corresponding one on $V$. For this
reason, we wish to lift $r\Delta$ to a second order differential operator on $V$
with rational functions as its coefficients. In order to do that, we first need
to  lift $\lambda_u$ to a rational function on $V$. Let $c_k(x)$ ($\tau_k(x)$
resp.) be the polynomial in $\tr x$, $\tr x^2$, \ldots, $\tr x^k$ such that, if
$x=\sum_{i=1}^k\lambda_ie_{ii}$, then
$$c_k(x)=\prod_{i=1}^k \lambda_i \quad \left(\tau_k(x)=\prod_{1\le i<j\le
k}(\lambda_i+\lambda_j)\quad \mbox{resp.}\right).
$$  For example, $c_1(x)=\tr x$, $\tau_1(x)=1$,
$c_2(x)={1\over 2}((\tr x)^2-\tr x^2)$ and $\tau_2(x)=\tr x$. Let
$D_k=k[1+(\rho-{k+1\over 2})\delta]$ --- the dimension of the canonical cone
of rank $k$, and
\begin{eqnarray}\label{potential}
\fbox{$\varphi_k =\tau_k^\delta\cdot c_k ^{\delta-1}\cdot  r^{2-D_k}$.}
\end{eqnarray}
For example, $D_1=1+(\rho-1)\delta$, and $\varphi_1=r^{-\delta(\rho -2)}$ up to a multiplicative constant. Note
that $\varphi_k$ is a rational function on $V$ and is positive on the canonical cone of rank $k$. From here one, we shall call $\varphi_k$ the {\bf phi-function} on the canonical cone of rank $k$.

With an orthonormal basis $\{e_\alpha\}$ for $V$ chosen, we write $x$ as $x^\alpha e_\alpha$, ${\partial \over \partial x^\alpha}$ as $\partial_\alpha$. Recall that ${/\hskip -5pt \partial}=\sum_\alpha e_\alpha \partial_\alpha$.
 
\begin{Prop}\label{rDelta} 
Fix a canonical cone of rank $k$. Let $\varphi$ be its phi-function, $\Delta$ be its Laplace operator. Then
 
 i) $\lambda_u$ can be lifted to a rational function on $V$:
 \begin{eqnarray}
\fbox{$4\lambda_u = \hat L_u\ln \varphi+\delta k\tr u$\, .}
\end{eqnarray}

ii) $r\Delta$ can be lifted to a second order differential operator on $V$ with
rational function coefficients:
\begin{eqnarray}
 \fbox{$r\Delta =\langle x\mid {/\hskip -5pt \partial}^2\rangle +
2\sum_\alpha \lambda_{e_\alpha}\partial_\alpha$\, .}
\end{eqnarray}

\end{Prop}

\begin{proof}

i) We just need to prove the identity at a point $x_0$ on the canonical cone. Choose a Jordan frame $\{e_{ii}\}$ such that
$x_0=\sum_{i=1}^k\lambda_i e_{ii}$ for some numbers $\lambda_1$, \ldots,
$\lambda_k$.  Then we extend $\{\sqrt \rho e_{ii}\}_{i=1}^\rho$ to an orthonomal basis
$\{e_\alpha\}$ such that $e_i=\sqrt{\rho}e_{ii}$, $1\le i\le k$, and $\{e_i\}_{i=1}^{D_k}$ is an basis of $\mbox{Im}L_{x_0}$.

In view of that fact that $x^k x^l=x^{k+l}$ and $\tr(x_0^k e_j)=0$ for $j>k$, by induction on $m$, we 
have
$$
\hat L_u(\tr x^m)|_{x_0}=-\sum_{i=1}^{k}\langle ux_0\mid e_i\rangle \partial_i
(\tr x^m)|_{x_0},
$$
consequently
\begin{eqnarray}\label{observation}
\hat L_u \varphi |_{x_0}=-\sum_{i=1}^k\langle ux_0\mid e_i\rangle \partial_i
\varphi |_{x_0}.
\end{eqnarray}
The rest of the proof is just a straightforward computation based on identity (\ref{observation}), so we leave it to readers.

ii) Since both sides are differential operators without the zero-th order term, in view of identity (\ref{KeyId1'}),
we just need to show that 
$$[\langle x\mid{/\hskip -5pt \partial}^2\rangle +2
\sum_\alpha\lambda_{e_\alpha}\partial_\alpha, \langle u\mid x\rangle ]=-2\tilde L_u,$$
something that can be easily verified. 
\end{proof}

For $\nu\in{\mathcal W}(V)\setminus\{0\}$, we introduce integer
\begin{eqnarray}
\rho(\nu)=\left\{
\begin{array}{ll}
k
&  \mbox{if $\nu= k{\delta \over 2}$},\\
\\
\rho &  \mbox{if $\nu> (\rho-1){\delta \over 2}$}.\\
\end{array}
\right.
\end{eqnarray}
and rational function
\begin{eqnarray}\label{phinu}
\varphi(\nu):=\left\{
\begin{array}{ll}
\varphi_k & \mbox{if $\nu=k{\delta\over 2}$ },\\
\\
\varphi_\rho\,\det(x)^{2\nu-\rho \delta } & \mbox{if  $\nu>(\rho-1){\delta \over
2}$ }
\end{array}
\right.
\end{eqnarray}
on $V$. Note that $\varphi(\nu)$ is always positive on the canonical cone of rank $\rho(\nu)$. 

For canonical cone $\mathcal C$, we let
$$
\acute D({\mathcal C}):= \{\psi:\; {\mathcal C} \to \bb C \mid \psi\in \acute D(V)\}, \quad \acute D_I({\mathcal C}):= \{\psi:\; {\mathcal C} \to \bb C \mid \psi\in \acute D_I(V)\}.
$$

\begin{Prop}\label{Lem4} Let $\nu\in{\mathcal W}(V)\setminus\{0\}$ and $\mathcal C$ be the canonical cone of rank $\rho(\nu)$. 
Then $\acute D({\mathcal C})$ is
dense in $L^2({\mathcal C}, {\sqrt{\varphi(\nu)}\over r}\mr{vol} )$.
\end{Prop}
\begin{proof}  Let us write $d\mu_\nu$ for ${\sqrt{\varphi(\nu)}\over r}\mr{vol}$.
Let $C_c({\mathcal C})$ be the set of compactly-supported continuous
complex-valued functions and 
$$M=\int_{\mathcal C} e^{-2r}\, d\mu_\nu.$$
It is clear that $M>0$.  By applying Theorem \ref{dmuk} in appendix A, one can
easily check that $M<\infty$.

Suppose that  $f\in L^2({\mathcal C},d\mu )$ and $\epsilon>0$. By Theorem 3.14
in Ref. \cite{Rudin87}, there is $g\in C_c({\mathcal C})$ such that
\begin{eqnarray}\label{estimate1}
||f-g||_{L^2}<{\epsilon\over 2}.
\end{eqnarray}
Since $e^r g\in C_c({\mathcal C})$, by the Stone-Weierstrass Theorem in Ref.
\cite{deBranges59}, there is a polynomial $p$ such that
\begin{eqnarray}\label{estimate2}
|e^r g - p|<{\epsilon\over 2\sqrt{M}} \quad\mbox{on ${\mathcal C}$},
\end{eqnarray} 
so
\begin{eqnarray}\label{estimate3}
||g-e^{-r}p||_{L^2}  &= &\left(\int_{\mathcal C}  |g-e^{-r}p|^2\,d\mu_\nu
\right)^{1\over 2}\cr
&< & { \epsilon\over 2} \left({1\over M}\int_{\mathcal C}  e^{-2r}\,d\mu_\nu
\right)^{1\over 2}\quad\mbox{using Eq. (\ref{estimate2})}\cr
&= &{\epsilon \over 2}.
\end{eqnarray}

Combining Eqs (\ref{estimate1}) and (\ref{estimate3}), we have
$$
||f-e^{-r}p||_{L^2}\le ||f-g||_{L^2}+||g-e^{-r}p||_{L^2}<\epsilon.
$$

\end{proof}

Let
\begin{eqnarray}
U(\nu):=\left\{
\begin{array}{ll}
{r\over 4}\left(\Delta(\ln \varphi(\nu))+{1\over 4}|d\ln \varphi(\nu)|^2\right)
&  \mbox{if $\nu\le (\rho-1){\delta \over 2}$}\\
\\
{r\over 4}\left(\Delta(\ln \varphi_\rho)+{1\over 4}|d\ln
\varphi_\rho|^2\right)\cr
+{\rho\over 4}\left((\nu-{n\over \rho})^2-({\delta \over 2}-1)^2\right) \tr
x^{-1} &  \mbox{if $\nu> (\rho-1){\delta \over 2}$}.\\
\end{array}
\right.
\end{eqnarray}
Here, $d$ and $|\;|$ denote the exterior derivative operator and the point-wise
norm for differential one-form on $\mathcal C$ respectively, $x^{-1}$ denotes the Jordan inverse of $x\in\Omega$. Note that $U(\nu)$
can be lifted to a rational function on $V$. Recall that $\tilde L_u=\hat
L_u-\lambda_u$.

\begin{Prop}\label{conjugation}Let $\nu\in{\mathcal W}(V)\setminus\{0\}$ and $\mathcal C$ be the canonical cone of rank $\rho(\nu)$. 

i) As a differential operator on $\mathcal C$,
\begin{eqnarray}
\tilde L_u = \sqrt[4]{\varphi_{\rho(\nu)}}\acute L_u(\nu) {1\over \sqrt[4]{\varphi_
{\rho(\nu)}}}.
\end{eqnarray}
 
ii) As a differential operator on $\mathcal C$,
\begin{eqnarray}
r\Delta = \,\sqrt[4]{\varphi(\nu)}(-i\acute X_e(\nu)){1\over
\sqrt[4]{\varphi(\nu)}}+U(\nu).
\end{eqnarray}

\end{Prop}
This proposition says that $\tilde L_u$ and $r\Delta$ are not as hard as they
might look. To prove this proposition, with the help of Proposition
\ref{rDelta}, one just needs to do some straightforward and relative short
computations, so we skip the proof. 

\subsection{The unitary realizations of TKK algebras on canonical cones}
Let $\nu\in{\mathcal W}(V)\setminus\{0\}$ and $\mathcal C$ a canonical cone of rank $\rho(\nu)$. Recall that $\varphi(\nu)$, a rational function introduced in Eq. (\ref{phinu}), is always positive on $\mathcal C$.  Upon recalling the definitions of $\acute D(\mathcal C)$ and $\acute D_I(\mathcal C)$ in the paragraph preceding to Proposition \ref{Lem4}, in view of Proposition \ref{conjugation}, we introduce
\begin{eqnarray}
 \tilde D({\mathcal C}) =\sqrt[4]{\varphi(\nu)}\,{\acute D}(\mathcal C), \quad
\tilde D_I({\mathcal C}) =\sqrt[4]{\varphi(\nu)}\,{\acute D}_I(\mathcal C) \nonumber
\end{eqnarray}
and differential operators with common domain $\tilde D({\mathcal C})$:
\begin{eqnarray}
{\tilde S}_{uv} (\nu) &=&\sqrt[4]{\varphi(\nu)}{\acute S}_{uv}(\nu){1\over
\sqrt[4]{\varphi(\nu)}} ,\cr
{\tilde X}_u (\nu) &=& \sqrt[4]{\varphi(\nu)}{\acute X}_u(\nu){1\over
\sqrt[4]{\varphi(\nu)}} , \cr
{\tilde Y}_v (\nu) &=& \sqrt[4]{\varphi(\nu)}{\acute Y}_v(\nu){1\over
\sqrt[4]{\varphi(\nu)}}.\nonumber
\end{eqnarray}
Note that these differential operators on $\mathcal C$ can be lifted to differential operators on $V$.

\begin{Prop}\label{PropDense}
Let $\nu\in{\mathcal W}\setminus\{0\}$ and $\mathcal C$ be a canonical cone of rank $\rho(\nu)$.

i) The TKK commutation relations (\ref{TKKRel}) hold under the replacement of $ O$ by $\tilde O(\nu)$.

ii) $\tilde D({\mathcal C})$ is a dense subset of $L^2\left(\mathcal C, {1\over
r}\mr{vol}\right)$.

iii)  ${\tilde S}_{uv}(\nu)$, ${\tilde  X}_u(\nu)$ and ${\tilde Y}_v(\nu)$ are
anti-hermitian operators on $\tilde  D({\mathcal C})$ with respect to hermitian inner product
$$
(\psi_1, \psi_2)=\int_{\mathcal C}\bar\psi_1\,\psi_2\, {1\over
r}\mr{vol}.
$$

iv) Let $\tilde {\ms D}_I({\mathcal C})$ be the orthogonal complement of $\tilde D_{I-1}({\mathcal C})$ in
$\tilde D_I({\mathcal C})$, then, under the unitary $\frk{u}$-action, we have the following orthogonal
decomposition
\begin{eqnarray}\label{Ktype}
\tilde D({\mathcal C})=\bigoplus_{I=0}^\infty \tilde{\ms D}_I({\mathcal C}).
\end{eqnarray}
Moreover, the finite dimensional vector space $\tilde{\ms D}_I({\mathcal C})$ is the eigenspace of $\tilde H_e:=i(\tilde
X_e+\tilde Y_e)$ with eigenvalue $(2I+\nu\rho)$.

v) Assume that ${\bf m}\in {\bb Z}^\rho$ with ${\bf m}\ge 0$ and $m_{\rho(\nu)+1}=0$. For and only for such $\bf m$, we let
$\tilde {\ms D}_{{\bf m}}({\mathcal C})$ be the orthogonal projection of $\sqrt[4]{\varphi(\nu)}e^{-r}{\mathcal P}_{\bf m}(V)$ onto $\tilde D_{|{\bf m}|}({\mathcal C})$. Then, as unitary representations of $ \frk u$, we have isomorphism $\tilde {\ms D}_{{\bf m}}({\mathcal C})\cong \xi_\nu\otimes {\mathcal P}_{ {\bf m}}(V)$ and orthogonal decomposition into irreducibles:
\begin{eqnarray}
\tilde {\ms  D}_I({\mathcal C})=\bigoplus_{{\bf m}\ge 0,  |{\bf m}|=I}^{m_{\rho(\nu)+1}=0} \tilde {\ms D}_{ {\bf m}}({\mathcal C}).
\end{eqnarray}

\end{Prop}
\begin{proof}

i) This quickly follows from Theorem \ref{universalQ}.

ii)  This  quickly follows from Proposition \ref{Lem4}.

iii) We start the proof with the following two observations: 1) multiplication
by a real-valued function is hermitian, herece $\tilde Y_v$ is anti-hermitian;
2) $r\Delta$ is hermitian, hence $\tilde X_e$ is anti-hermitian in view of part
ii) of Proposition \ref{conjugation}. Combining these observations with the commutation relations
in part i), the proof follows quickly.

iv) The orthogonal decomposition follows from the following two facts: 1) the $\frk{u}$-action is
unitary, a fact from part iii)  above, 2)  the $\frk{u}$-action commutes with the inclusion of $D_{I-1}({\mathcal C})$ into
$D_I({\mathcal C})$, a fact implied by part i) of Theorem \ref{universalQ2}.  The remaining part follows
from the fact that $\tilde H_e$ is hermitian and part iii) of
Theorem \ref{universalQ2}.

v) This follows from part ii) of Theorem \ref{universalQ2}.
\end{proof}

\begin{rmk}
In view of Proposition \ref{PropDense}, the semi-positive hermitian form $(,)_\nu$ mentioned in the beginning paragraph of this section is
\begin{eqnarray}\label{hermitionformgeneral}
(\psi_1, \psi_2)_\nu=\int_{\mathcal C}\bar\psi_1\, \psi_2\, {\sqrt{\varphi(\nu)}\over
r}\mr{vol}
\end{eqnarray}
and the space of ``spurious states" consists
of functions in $\acute D(V)$ which vanish on $\mathcal C$. One can check that Eq. (\ref{hermitionformgeneral})  is a generalization of Eq. (\ref{hermitionform}) and ${\acute D}(\mathcal C)$ is the quotient of ${\acute D}(V)$ by the 
the space of ``spurious states". Therefore, the measure $d\mu_\nu$ mentioned in Subsection \ref{SS: QTKK} is equal to
${\sqrt{\varphi(\nu)}\over r}\mr{vol}$. When $\nu>(\rho-1){\delta\over 2}$, up to a multiplicative constant, this explicit formula for $d\mu_\nu$ agrees with the one retrieved from the bottom line of page 271 in Ref. \cite{FK94}. Of course, our explicit formula
works even if $\nu$ takes a discrete value $k{\delta\over 2}$, $1\le k< \rho$. 

\end{rmk}

Denote by $\tilde{\mr U}(V)$ (or simply $\tilde{\mr U}$) the simply connected Lie group whose Lie algebra is $\frk{u}(V)$. Recall that $\mathcal W(V)$ is the Wallach set of $V$, $\lambda_0$ is the fundamental weight conjugate to the unique non-compact
simple root $\alpha_0$ in the simple root system for $\frk{co}$ in Lemma \ref{LemmaVogan}. 

\begin{Thm} Let $V$ be a simple euclidean Jordan algebra, $\nu\in\mathcal W(V)\setminus\{0\}$,  and $\mathcal C$ be $V$'s canonical cone of rank $\rho(\nu)$. Under the action of $\frk{co}(V)$ which maps $O$ to $\tilde 
O(\nu)$, $\tilde D({\mathcal C})$ becomes a unitary lowest weight $(\frk{co}(V), \tilde{\mr U}(V))$-module with lowest
weight $\nu \lambda_0$ and has the following multiplicity free $K$-type formula \footnote{I.e., the formulae for decomposing
${\tilde D}({\mathcal C})$ (considered as a $ \tilde{\mr U}(V)$-module) into its irreducible components. ``Multiplicity free" means each irreducible component appears only once.}:
\begin{eqnarray}
{\tilde D}({\mathcal C}) = \bigoplus_{{\bf m}\ge 0}^{ m_{\rho(\nu)+1}=0} \tilde {\ms D}_{{\bf m}}({\mathcal C}).
\end{eqnarray}
Therefore, as a representation of $\tilde{\mr U}(V)$, ${\tilde D}({\mathcal C}) \cong \bigoplus_{{\bf m}\ge 0}^{m_{\rho(\nu)+1}=0} \xi_\nu\otimes {\mathcal P}_{ {\bf m}}(V)$.

Consequently, upon integration, $L^2\left(\mathcal C, {1\over
r}\mr{vol}\right)$ becomes a unitary lowest weight representation $\pi_\nu$ for $\mr{Co}(V)$
with the same lowest weight.
\end{Thm}

\begin{proof} Parts i), iii) and iv) of Proposition \ref{PropDense} imply that $\tilde D({\mathcal C})$ is
a unitary $(\frk{co}(V), \tilde{\mr U}(V))$-module. The $K$-type formula follows from parts iv) and v) of Proposition \ref{PropDense}.
Combining with parts v) and  vi) of Theorem \ref{universalQ2}, we arrive at first part of this theorem. The second part follows from the first part, part ii) of
Proposition \ref{PropDense},  and a fundamental theorem of Harish-Chandra.

\end{proof}
In view of the classification theorem in Ref. \cite{EHW82}, the nontrivial scalar-type unitary lowest weight representations of $\mr{Co}(V)$ are exhausted by representations $\pi_\nu$ in the above theorem.

\section{Generalized Quantum Kepler Problems Without Magnetic Charges}\label{S: GKP}
In Ref. \cite{meng10}, we introduce the universal hamiltonian for the Kepler
problem in terms of the generators of TKK algebra, and remark that whenever we
have a quantization for the TKK algebra, we have a super-integrable model of the
Kepler-type. In view of the quantizations for the TKK algebra presented in the
last section, we have some new super-integrable models of Kepler-type.

As before, $V$ is a simple euclidean Jordan algebra of rank
$\rho$ and degree $\delta$, $\mathcal W(V)$ is its Wallach set. For a canonical cone inside $V$ of rank $k$, we use
$\varphi_k$ to denote the phi-function defined in Eq. (\ref{potential}) and $\Delta$ to denote its (non-positive) Laplace operator. For $\nu\in \mathcal W(V)\setminus\{0\}$, we let
\begin{eqnarray}
V(\nu):=\left\{
\begin{array}{ll}
{1\over 8}\left(\Delta(\ln \varphi_k)+{1\over 4}|d\ln \varphi_k|^2\right)
&  \mbox{if $\nu=k{\delta \over 2}$},\\
\\
{1\over 8}\left(\Delta(\ln \varphi_\rho)+{1\over 4}|d\ln
\varphi_\rho|^2\right)\cr
+{\rho\over 8}\left((\nu-{n\over \rho})^2-({\delta \over 2}-1)^2\right) {\tr
x^{-1}\over r} &  \mbox{if $\nu> (\rho-1){\delta \over 2}$}.\\
\end{array}
\right.
\end{eqnarray}
and call $V(\nu)$ the {\bf quantum-correction potential} on the canonical cone of rank $\rho(\nu)$. Note that $V(\nu)={U(\nu)\over 2r}$. Here is the definition of the generalized quantum Kepler problem attached to $\pi_\nu$:

\begin{Definition}[Generalized Quantum Kepler Problems]  Let $V$ be a simple euclidean Jordan algebra and $\nu\in {\mathcal W}(V)\setminus\{0\}$. The $\nu$-th generalized quantum Kepler problem of $V$ is the quantum mechanical system for which the configuration space is the canonical cone of rank $\rho(\nu)$, and the
hamiltonian $\tilde H(\nu)$ (or simply $\tilde H$) is
\begin{eqnarray}
-{1\over 2} \Delta + V(\nu)-{1\over r}.
\end{eqnarray}
Here, $\Delta$ and $V(\nu)$ are the Laplace operator and the quantum-correction potential respectively. 
\end{Definition}
One can verify that when $\nu={\delta\over 2}$, generalized quantum Kepler problem  is the $J$-Kepler problem in Ref. \cite{meng09}, and to get the original Kepler problem we need to take $V=\Gamma(3)$ and $\nu=1$.

\subsection{Solution of the bound state problem}
Given a generalized quantum Kepler problem on a canonical cone $\mathcal C$, the bound state problem is primarily the following (energy) spectrum problem:
\begin{eqnarray}\label{eigen}
\left\{\begin{array}{rcl}
\tilde H\psi & = & E\psi\\
\\
\displaystyle\int_{\mathcal C} |\psi|^2\, \mr{vol}&< & \infty, \quad \psi\not\equiv 0.
\end{array}\right.
\end{eqnarray}
It turns out that $E$ has to take certain discrete values. For
example, for the original Kepler problem, we have
$E=-{1\over 2n^2}$, $n=1, 2, \ldots$

We shall use $\ms H_I$ to denote the $I$-th {\bf energy eigenspace}, $I=0, 1, \ldots$ and $\ms H$ to denote the {\bf Hilbert space of bound states} --- the $L^2$-completion of $\bigoplus_{I=0}^\infty{\ms H}_I$. 
\begin{Thm}\label{main3}  Let $V$ be a simple euclidean Jordan algebra and $\nu\in {\mathcal W}(V)\setminus\{0\}$. 
For the $\nu$-th generalized quantum Kepler problem of $V$, the following statements are true:

i) The bound state energy spectrum is
$$
E_I=-{1/2\over (I+\nu{\rho\over 2})^2}
$$ where $I=0$, $1$, $2$, \ldots

ii) As a representation of $\tilde{\mr U}(V)$, ${\ms H}_I\cong \bigoplus_{{\bf m}\ge 0, |{\bf m}|=I} ^{m_{\rho(\nu)+1}=0} \xi_\nu\otimes{\mathcal P}_{\bf m}(V)$.

iii) $\ms H$ provides a realization for representation $\pi_\nu$.
\end{Thm}
\begin{proof}
In view of part iv) of Proposition \ref{PropDense}, we start with the eigenvalue problem for $-{1\over 2}\tilde  H_e$:
\begin{eqnarray}
-{1\over 2}\tilde H_e\tilde \psi=-n_I \tilde \psi,
\end{eqnarray} where $n_I=I+\nu{\rho\over 2}$ and $\tilde\psi\not\equiv 0$ is square integrable
with respect to measure ${1\over r}\mr{vol}$ on the canonical cone of rank $\rho(\nu)$. The above equation can be recast as
$$
-{1\over 2}\left(\Delta-{U(\nu)\over r}+{2n_I\over
r}\right)\tilde\psi(x) =-{1\over 2}\tilde \psi(x).
$$
Let $\psi(x):=\tilde \psi({x\over n_I})$, then the preceding
equation becomes
$$
\left(-{1\over 2}\Delta+V(\nu)-{1\over r}
\right)\psi(x)=-{1/2\over n_I^2}\psi(x),$$
i.e.,
\begin{eqnarray}
\tilde H \psi =-{1/2\over n_I^2}\psi.
\end{eqnarray}
One can check that $\psi$ is square integrable
with respect to measure $\mr{vol}$. Therefore, $\tilde\psi$ is an eigenfunction of $\tilde H_e$ with eigenvalue $2n_I$
$\Rightarrow$ $\psi$ is an eigenfunction of $\tilde H$ with eigenvalue $-{1/2\over n_I^2}$. By turning the above arguments backward, with the help of an explicit form of the eigenfunctions for $\tilde H$, one can show that the converse of this statement is also true. Therefore, in view of parts iv) and v) of Proposition \ref{PropDense}, we have
 $${\ms H}_I\cong \tilde {\ms D}_I(\mathcal C)\cong \bigoplus_{{\bf m}\ge 0, |{\bf m}|=I}^{m_{\rho(\nu)+1}=0}  \xi_\nu\otimes{\mathcal P}_{\bf m}(V).$$

Introduce 
\begin{eqnarray}
\begin{array}{lccc}
\tau:& \bigoplus_{I=0}^\infty{\ms H}_I & \longleftarrow& {\tilde D}({\mathcal C}) = \bigoplus_{{\bf m}\ge 0}^{m_{\rho(\nu)+1}=0} \tilde {\ms D}_{{\bf m}}({\mathcal C})\\
\\
& c_{\bf m}\tilde \psi_{\bf m}({x\over n_{|{\bf m}|}})
& \longleftarrow \hskip -4pt |& \tilde \psi_{\bf m}(x)\in \tilde {\ms D}_{{\bf m}}({\mathcal C})
\end{array}
\end{eqnarray}
Here $c_{\bf m}$ is a constant depending on ${\bf m}$. The value of $c_{\bf m}$ can be determined and $\tau$ can be shown to be an isometry, provided that an analogue of Theorem 2 in Ref. \cite{CD2003} for generalized Laguerre polynomials
can be established, something that definitely can be done.  Since ${\tilde D}({\mathcal C}) $ is a unitary highest weight Harish-Chandra module, and $\tau$ is an isometry, $\bigoplus_{I=0}^\infty{\ms H}_I$ becomes a unitary highest weight Harish-Chandra module. Since the $L^2$-completion of $\bigoplus_{I=0}^\infty{\ms H}_I$ is the
Hilbert space of bound states, we arrive at part iii) of this theorem.
\end{proof}

Generalized Kepler problems are natural generalizations of the J-Kepler problems, but with an important difference: the energy eigenspaces are no longer always irreducible representations of $\tilde{\mr U}(V)$, cf. part ii) of the theorem above.

\subsection{An alternative route to the generalized quantum Kepler problems}
A main objective of this paper is to construct the generalized quantum Kepler problems. The route we have taken is technically easy, but perhaps less intuitive. The purpose here is to point out that there is a more intuitive, though technically harder, route towards the generalized quantum Kepler problems. 

Let us start with the classical universal hamiltonian 
$$
{\mathcal H}={1\over 2}{\langle \pi\mid L_x \mid \pi\rangle\over r}-{1\over r}
$$ on $V$. Although $(V,ds^2_E)$ is a Riemannian manifold, the first term in $\mathcal H$ is not the kinetic energy for a single particle moving on $V\setminus\{0\}$. A little experiment will convince the readers that the only way out is to replace $V\setminus\{0\}$ by a canonical cone.  In that way, we get $\rho$ generalized classical Kepler problems, exactly one
on each canonical cone:
\begin{Definition}[Generalized Classical Kepler Problems]  Let $V$ be a simple euclidean Jordan algebra and $k$ be an integer between $1$ and $\rho$. The $k$-th generalized classical Kepler problem of $V$ is the classical mechanical system for which the configuration space is the canonical cone of rank $k$, and the Lagrangian
 is
\begin{eqnarray}
L(x,\dot x)={1\over 2} ||\dot x||^2+{1\over r}.
\end{eqnarray} 
\end{Definition}
By quantizing the generalized classical Kepler problems, we expect to get the generalized quantum Kepler problems. However, due to the operator ordering problem, in general we don't know how to do it! The way to overcome this technical hurdle is to demand that the hidden symmetry in the generalized classical Kepler problems
be still present after quantization.  But that really leads us back to the issue of quantizing the TKK algebra.

\appendix
\section{Polar coordinates}\label{App: A}
The purpose of this section is to understand the polar coordinates on ${\mathcal
C}_k$. The theorem obtained here is an extension of Theorem VI.2.3 in Ref. \cite{FK94} from symmetric cones to canonical cones and the presentation follows that of Section 2 of Chapter VI in Ref. \cite{FK94}. 

We fix a Jordan frame: $e_{11}$, \ldots, $e_{\rho\rho}$ and a Jordan
basis $\{e_{ii}, e_{ij}^\mu\}$. We denote by $V_{ij}$ the corresponding $(i,j)$-Peirce component of $V$.  Let
$$
R_k=\left\{\sum_{i=1}^k a_ie_{ii} \mid a_i\in \bb R \right\}, \quad
R_k^+=\left\{ \sum_{i=1}^k a_ie_{ii} \mid  a_1> a_2> \cdots > a_k>0 \right\}
$$
Let $K$ be the identity component of $\mr{Aut}(V)$ and  $M_k$ be the subgroup $K$ fixing each point $a\in R_k$:
$$
M_k=\{g\in K \mid \forall a\in R_k, ga=a \}
$$
and $m_k$ be its Lie algebra:
$$
{\frk m}_k=\{X\in \frk{der} \mid  \forall a\in R_k, Xa=0 \}$$

For $i< j$, we define
$$
\frk l_{ij}=\{[L_{e_{ii}}, L_\xi]\mid \xi \in V_{ij}\}
$$
Let
$$
{\frk l}_k=\sum_{1\le i\le k, i<j}{\frk l}_{ij}
$$

\begin{Prop}
Let $a=\sum_{i=1}^ka_ie_{ii}\in R_k^+$. For $X\in \frk{der}$, $(a, Xa)\in
T_a{\mathcal C}_k$ is orthogonal to $T_aR_k^+$, and if $a_i\neq a_j$ for $i\neq
j$, the map
\begin{eqnarray}\label{appendixiso}
{\frk l}_k &\to & (T_aR_k^+)^\perp\cr
X&\mapsto& (a, Xa)
\end{eqnarray}
is an isomorphism.
\end{Prop}
\begin{proof}
For $X=[L_u, L_v]$, $u$ and $v$ in $V$, and for $a\in R_k^+$ and $(x, b)$ in
$T_aR_k^+$:
\begin{eqnarray}
((a,Xa), (a,b)) &=& \langle a\mid e\rangle \langle Xa\mid \bar L_a^{-1} b\rangle
\cr
&= &  \langle a\mid e\rangle \langle [L_a, L_{ \bar L_a^{-1} b}]u\mid v\rangle
\cr
&=& 0
\end{eqnarray}
because both $L_a$ and $L_{\bar L_a^{-1} b}$, being of diagonal form with respect to
the Jordan basis, commute with each other.

Assume that $1\le i\le k$ and $i<j\le \rho$. For $X$ in ${\frk l}_{ij}$:
$$
X=[L_{e_{ii}}, L_\xi], \quad \xi\in V_{ij},
$$
and for $a=\sum_{i=1}^k a_ie_{ii}$ in $R_k^+$ we have
$$
Xa={1\over 4}(a_j-a_i)\xi.
$$
Here it is understood that $a_j=0$ if $j>k$.
Therefore, the range of the map $X\mapsto Xa$ contains the subspaces $V_{ij}$
and the sum
$$
\bigoplus_{1\le i\le k, i<j}V_{ij}.
$$
Since $(T_aR_k^+)^\perp=\{a\}\times \bigoplus_{1\le i\le k, i<j}V_{ij}$,  map
(\ref{appendixiso}) is onto, hence must be an isomorphism because the dimensions of the domain and the target are equal.
\end{proof}

\begin{Cor}
---i) As a vector space, we have
$$
{\frk {der}}={\frk m}_k\oplus {\frk l}_k.
$$
ii) Map
\begin{eqnarray}
\phi:\quad K/M_k \times R_k^+ &\to& {\mathcal C}_k\cr
(gM_k, a)&\mapsto& ga\nonumber
\end{eqnarray}
has dense range and is a diffeomorphism onto its range.
\end{Cor}

\begin{Th}\label{dmuk} Write $d\mu_{k{\delta\over 2}}$ for ${\sqrt{\varphi_k}\over r}\mr{vol}$.
Under the identification map $\phi$, we have
\begin{eqnarray}
d\mu_{k{\delta\over 2}}= C\,\mr{vol}_{K/M_k}\prod_{1\le i<j\le k}(a_i-a_j)^\delta \prod_{i=1}^k
\left(a_i^{{\delta \over 2}(\rho-k+1)-1} da_i \right)
\end{eqnarray}
where $\mr{vol}_{K/M_k}$ is the $K$-invariant measure on $K/M_k$ and $C$ is a
constant  depending only on ${\mathcal C}_k$.
\end{Th}
\begin{proof}
We start with a local parametrization
of $\mathcal C_k$ around point $a\in R_k^+$:
$$
x=\exp{\left(\sum_{1\le i\le k, i<j\le \rho}^{1\le\alpha\le \delta
}x_{ij}^\alpha X_{ij}^\alpha\right)} a.
$$
Here, \fbox{$X_{ij}^\alpha=[L_{e_{ii}}, L_{e_{ij}^\alpha}]$}. Then
$$
dx|_a=\sum_{i=1}^ke_{ii}da_i+\sum_{1\le i\le k, i<j\le \rho}^{1\le\alpha\le
\delta }{1\over 4}(a_j-a_i)e_{ij}^\alpha dx_{ij}^\alpha\mid_{x_{ij}^\alpha=0},$$
where it is understood that $a_j=0$ if $j>k$. Therefore, we can calculate the canonical metric $ds^2_K$ at point $a$:
\begin{eqnarray}
ds_K^2|_a &= &\langle a\mid e\rangle
\langle dx\mid \bar L_a^{-1}dx\rangle\mid_a\cr
&=& {\langle a\mid e\rangle\over \rho}\left(\sum_i {1\over a_i}da_i^2+{1\over
8}\sum_{1\le i\le k, i<j\le \rho}^{1\le\alpha\le \delta } {(a_j-a_i)^2\over
a_i+a_j}(dx_{ij}^\alpha)^2\mid_{x_{ij}^\alpha=0}\right),
\end{eqnarray}
so, up to a multiplicative numerical constant, we have
$$
\mr{vol}|_a=  (r^{D_k/2}c_k^{{1\over 2}(\delta (\rho-k)-1)}\tau_k^{-{\delta \over
2}}) |_a\prod_{1\le i<j\le k}(a_i-a_j)^\delta \bigwedge_{i=1}^k
da_i\bigwedge\bigwedge_{1\le i\le k, 1\le\alpha\le \delta}^{ i<j\le \rho}
dx_{ij}^\alpha|_{x_{ij}^\alpha=0}
$$
and 
$$
d\mu_{k{\delta\over 2}} |_a=  c_k(a)^{{\delta \over 2}(\rho-k+1)-1}\prod_{1\le i<j\le
k}(a_i-a_j)^\delta \bigwedge_{i=1}^k da_i\bigwedge\bigwedge_{1\le i\le
k, 1\le\alpha\le \delta}^{i<j\le \rho} dx_{ij}^\alpha|_{x_{ij}^\alpha=0}.
$$
On the other hand, since $K$ is a simple Lie group, one can show that $X_{ij}^\alpha$'s are mutually orthogonal
with respect to the negative-definite Cartan-Killing form\footnote{One just needs to show that
the trace of
$X_{ij}^\alpha X_{i'j'}^{\alpha'}$ is zero if $(i,j,\alpha)\neq
(i',j',\alpha')$.}, so the $K$-invariant volume form on $K/M_k$ at $eM_k$ is
equal to $\bigwedge_{1\le i\le k, 1\le\alpha\le \delta}^{i<j\le \rho}
dx_{ij}^\alpha|_{x_{ij}^\alpha=0}$ modulo a multiplicative numerical constant. 
Since $d\mu_k$ is also
$K$-invariant, up to a multiplicative numerical constant, we have
$$
d\mu_{k{\delta\over 2}}= \mr{vol}_{K/M_k}\prod_{1\le i<j\le k}(a_i-a_j)^\delta \prod_{i=1}^k
\left(a_i^{{\delta \over 2}(\rho-k+1)-1} da_i \right)$$
as a measure. 
\end{proof}
As a side remark, we would like to mention the fact that integral
$$
\int_\Omega e^{-2r}\det (x)^{\nu-\rho{\delta\over 2}}\, d\mu_{\rho{\delta\over 2}}
$$
is finite if and only if $\nu>(\rho-1){\delta\over 2}$.

\section{List of notations}\label{App:B}
The purpose here is to list some basic notations and terminologies for
this paper and its sequels. 

\begin{itemize}
\item $V$ --- a (finite dimensional) simple euclidean Jordan algebra;
\item $e$, $\rho$, $\delta$, and $n$ --- reserved for the identity element,
rank, degree, and dimension of $V$;
\item $\tr u$, $\det u$ --- the trace, determinant of $u\in V$; 
\item $\langle u\mid v\rangle$ --- the inner product of $u, v\in V$, and is
chosen to be ${1\over \rho}\tr(uv)$;
\item $x$ --- reserved for a generic point in $V$ when $V$ is considered as a
smooth space;
\item $r$ --- reserved for function $\langle e\mid \;\rangle$ on smooth space
$V$;
\item $\{e_\alpha\}$ --- an orthonormal basis for $V$;
\item $x^\alpha$ --- the coordinates of $x\in V$ with respect to basis
$\{e_\alpha\}$;
\item $\pi$ ---  reserved for a generic point in $V$ when $V$ is considered as
the tangent space of $V$;
\item $\pi^\alpha$ --- the coordinates of $\pi\in V$ with respect to basis
$\{e_\alpha\}$;
\item ${/\hskip -5pt\partial}$ --- a shorthand notation for $\sum_\alpha e_\alpha {\partial \over \partial x^\alpha}$;
\item ${\backslash \hskip -6pt \partial}$ --- a shorthand notation for $\sum_\alpha e_\alpha {\partial \over \partial \pi^\alpha}$;
\item $d$ --- the exterior derivative operator;
\item $\mr{vol}$ --- the volume form;
\item $uv$ --- the Jordan product of $u, v\in V$;
\item $\{uvw\} $--- the Jordan triple product of $u,v,w\in V$;
\item $L_u$ --- the multiplication by $u\in V$;
\item $S_{uv}$ --- defined to be $[L_u, L_v]+L_{uv}$, so $S_{uv}w=\{uvw\}$;  
\item ${\mathcal W}(V)$ --- the Wallach set of $V$;
\item ${\mathcal P}(V)$ --- the set of complex-valued polynomial functions on $V$;
\item $\frk{der}(V)$, $\frk{der}$ --- the derivation algebra of $V$;
\item $\frk{str}(V)$, $\frk{str}$ --- the structure algebra of $V$, it is generated by $L_u$, $u\in V$;
\item $\frk{co}(V)$, $\frk{co}$ --- the conformal algebra of $V$;
\item $\frk{u}(V)$, $\frk{u}$ --- the maximal compact Lie subalgebra of $\frk{co}$;
\item $\mr{Aut}(V)$ --- the automorphism group of $V$;
\item $\mr{Str}(V)$, $\mr{Str}$ --- the structure group of $V$;
\item $\mr{Co}(V)$, $\mr{Co}$ --- the conformal group of $V$, and is defined to be the simply
connected Lie group with $\frk{co}$ as its Lie algebra;
\item ${\tilde{\mr U}}(V)$, $\tilde{\mr U}$--- the simply connected Lie group with $\frk{u}$ as its Lie
algebra;
\item ${\tilde H}(\nu)$, $\tilde H$ --- the hamiltonian of the generalized Kepler problem corresponding to Wallach parameter $\nu$;
\item ${\ms H}_I$ --- the $I$th energy eigenspace for $\tilde H$;
\item $\ms H$ --- the Hilbert space of bound states for $\tilde H$. 
\end{itemize}


\begin{thebibliography}{99}

\bibitem{PJordan33}
P. Jordan, {\em Z. Phys.} {\bf 80} (1933), 285.

\bibitem{JVW34} P. Jordan, J. von Neumann and E. P. Wigner,
{\em Ann. Math.} {\bf 35} (1934), 29.

\bibitem{K. McCrimmon2004} K. McCrimmon, {\em A taste of Jordan algebras},
Universitext, Springer-Verlag, New York, 2004.

\bibitem{FK94}
J.Faraut and A.Kor\'{a}nyi, {\em Analysis on Symmetric Cones}, Oxford
Mathematical Monographs, 1994.

\bibitem{TKK60} J. Tits, Nederl. Akad. van Wetens. {\bf 65} (1962), 530; M.
Koecher, Amer.
J. Math. {\bf 89} (1967) 787; I.L. Kantor, Sov. Math. Dok. {\bf 5}
(1964), 1404.

\bibitem{MG93} M. G\"{u}naydin, {\em Mod. Phys. Lett. A} {\bf 8} (1993), 1407-1416. 


\bibitem{ADO06}
M. Aristidou, M. Davidson and G. \'{O}lafsson, {\em Bull. Sci. math. }{\bf 130} (2006), 246-263. 

\bibitem{meng09}
G. W. Meng, Euclidean Jordan Algebras, Hidden Actions, and $J$-Kepler Problems.
{\em ArXiv}: 0911.2977 [math-ph]



\bibitem{meng10}  
G. W. Meng,  The Universal Kepler Problem. 
{\em ArXiv}:1011.6609 [math-ph]



\bibitem{EHW82} T. Enright, R. Howe and N. Wallach, {\em
Representation theory of reductive groups}, Progress in Math. {\bf
40}, Birkh\"{a}user (1983), 97-143; H. P. Jakobsen,  {\em J. Funct. Anal.} {\bf 52} (1983), no. 3,
385-412.

\bibitem{Wallach79} 
N. Wallach, {\em Trans. Amer. Math. Soc.} {\bf 25} (1979), 1-17, 19-37. 

\bibitem{Vergne&Rossi76} 
M. Vergne and H. Rossi, {\em  Acta Math.}  {\bf 136}  (1976), no. 1-2, 1-59.


\bibitem{Dvorsky&Sahi03}
A. Dvorsky and S. Sahi, {\em J. Funct. Anal.} {\bf 201} (2003), no. 2, 430-456.


\bibitem{CD2003}
C. Dunkl,  {\em Analysis and Applications}, Vol. {\bf 1}, No. 2 (2003) 177-188.


\bibitem{BarutKleinert67}
A.O. Barut and H. Kleinert, {\em Phys. Rev.} {\bf 156} (1967), 1541-1545.


\bibitem{Rudin87}
W. Rudin, {\em Real and complex analysis - 3rd edition}, McGraw-Hill Book Co.,
New York, 1987. 

\bibitem{deBranges59}
L. de Branges,   {\em Proc. Amer. Math. Soc.}  {\bf 10} (1959), 822�824. 


\bibitem{IBars10}
I. Bars, {\em Int. J. Mod. Phys. A} {\bf 25} (2010), 5235-5252. 



\end{thebibliography}
\end{document}